%% file: Draft_Batch_codes.tex
\documentclass[11pt,twoside]{article}
\usepackage[margin=0.9in]{geometry}

\usepackage{url,hyperref,xspace,fullpage}
\usepackage[usenames]{color}
\usepackage{times}
\usepackage{verbatim}
\usepackage{algorithm}
\usepackage{algorithmic}
\usepackage{amsmath}
\usepackage{amssymb}
\usepackage{amsthm}
\usepackage{dsfont}
\usepackage{color}
\usepackage{colortbl}
\usepackage{float}
\usepackage{cite}
\usepackage{dsfont}
\usepackage{bm}
\usepackage{authblk}
\usepackage[pdftex]{graphicx}
\usepackage{caption}
\usepackage{subcaption}

\usepackage{url, hyperref}

\theoremstyle{plain}
\newtheorem{theorem}{Theorem}
\newtheorem{proposition}{Proposition}
\newtheorem{lemma}{Lemma}
\newtheorem{definition}{Definition}

\theoremstyle{definition}
\newtheorem{example}{Example}

\theoremstyle{remark}
\newtheorem{remark}{Remark}

\newcommand{\beq}{\begin{eqnarray}}
\newcommand{\eeq}{\end{eqnarray}}

\newcommand{\field}[1]{\mathbb{#1}}

\newcommand{\F}{\field{F}}

\newfont{\bbb}{msbm10 scaled 500}

\newfont{\bb}{msbm10 scaled 1100}

\newcommand{\NN}{\mbox{\bb N}}

\newcommand{\cv}{{\bf c}}

\newcommand{\ev}{{\bf e}}

\newcommand{\lv}{{\bf l}}

\newcommand{\vv}{{\bf v}}
\newcommand{\xv}{{\bf x}}
\newcommand{\yv}{{\bf y}}

\newcommand{\Bc}{{\cal B}}
\newcommand{\Cc}{{\cal C}}

\newcommand{\Ec}{{\cal E}}

\newcommand{\Gc}{{\cal G}}

\newcommand{\Nc}{{\cal N}}

\newcommand{\Pc}{{\cal P}}

\newcommand{\Sc}{{\cal S}}
\newcommand{\Tc}{{\cal T}}

\newcommand{\Vc}{{\cal V}}

\definecolor{OXO-emph}{RGB}{153,0,0}

\newcommand\ceilb[1]{\left\lceil #1 \right\rceil}
\newcommand\floorb[1]{\left\lfloor #1 \right\rfloor}
\newcommand{\algrule}[1][.2pt]{\par\vskip.5\baselineskip\hrule height #1\par\vskip.5\baselineskip}
\DeclareMathAlphabet{\mathpzc}{OT1}{pzc}{m}{it}

\renewenvironment{proof}[1][\proofname]{{\bfseries #1.}}{\qed}

\newcounter{ALC@tempcntr}
\newcommand{\LCOMMENT}[1]{%
    \setcounter{ALC@tempcntr}{\arabic{ALC@rem}}
    \setcounter{ALC@rem}{1}
     \hfill~~~~~~~~~~ ~~~~~~~~~{\bf /*~~#1 */}
    \setcounter{ALC@rem}{\arabic{ALC@tempcntr}}
}%

\title{Batch Codes through Dense Graphs without Short Cycles}

\author[1]{Alexandros~G.~Dimakis$^\P$\thanks{A.~G. Dimakis would like to acknowledge support from grants NSF CCF-1422549, CCF-1344364, and CCF-1344179.}}
\author[2]{Anna G\'{a}l$^\S$\thanks{A.~G\'{a}l would like to acknowledge support from grant NSF CCF-1018060.}} 
\author[1]{Ankit Singh Rawat$^\sharp$} 
\author[2]{Zhao Song$^\ddagger$} 
\affil[1]{Department of Electrical and Computer Engineering, The University of Texas at Austin, TX \, {$^\P$dimakis@austin.utexas.edu \quad $^\sharp$ankitsr@utexas.edu}}
\affil[2]{Department of Computer Science, The University of Texas at Austin, TX \quad \quad \quad \quad \quad \quad{$^\S$panni@cs.utexas.edu \quad $^\ddagger$zhaos@utexas.edu} }

\begin{document}

\maketitle

\begin{abstract}

Consider a large database of $n$ data items that need to be stored using $m$ servers.
We study how to encode information so that a large number $k$ of read requests can be performed \textit{in parallel} while the rate remains constant (and ideally approaches one). 
This problem is equivalent to the design of multiset Batch Codes introduced by Ishai, Kushilevitz, Ostrovsky and Sahai~\cite{batch}.

We give families of multiset batch codes
 with asymptotically optimal rates of the form $1-1/\text{poly}(k)$ 
and a number of servers $m$ scaling polynomially in the number of read 
requests $k$. 
An advantage of our batch code constructions over most previously
known multiset batch codes is explicit and deterministic decoding 
algorithms and asymptotically optimal fault tolerance.
 
Our main technical innovation is a graph-theoretic method of designing 
multiset batch codes using dense bipartite graphs with no small cycles. 
We modify prior graph constructions of dense, high-girth graphs 
to obtain our batch code results. 
We achieve close to optimal tradeoffs between the parameters
for bipartite graph based batch codes.

\end{abstract}

\thispagestyle{empty}



\input{intro.tex}

\input{organization.tex}
\input{background.tex}

\input{main_results.tex}

\input{constructions.tex}
\input{constructions2.tex}

\input{zigzag_construction.tex}

\input{reliability.tex}


\bibliographystyle{abbrv}
\bibliography{batch_codes}

\newpage


\end{document}

%% file: intro.tex
\section{Introduction}
Batch codes were introduced by Ishai, Kushilevitz, Ostrovsky and Sahai 
in~\cite{batch}, motivated by applications to load balancing in distributed storage  and private information retrieval.
An $(n,N,k,m)$ batch code encodes a string $\xv$ of $n$ symbols to a string $\yv$ 
of $N$ symbols to be stored at  
$m$ distinct servers (also called buckets), 
so that any $k$ read requests for symbols of $\xv$ can be performed by 
reading at most $1$ (more generally $t$) symbols from each bucket.

In this paper we study batch codes in the context of their applications in
distributed storage.
Consider a large database of $n$ data items that need to be stored using 
$m$ servers. 
The simplest and most frequently used way of introducing redundancy in distributed storage systems is replication. Assume that the whole database is replicated $k$ times using $m=k$ servers.
One advantage of $k$-replication is that \textit{any} $k$ read requests can be processed in parallel\footnote{We assume that each server can serve only one read request at a time, but this can be easily generalized.} This includes $k$ requests to read one item, $k/2$ requests to read any two items, \textit{etc.} Further, replication provides fault tolerance since any $k-1$ server failures can be tolerated with no data loss. The disadvantage, of course, is that the rate of this storage scheme is vanishing as $1/k$. This large storage overhead of replication is a major cost bottleneck for large-scale storage clusters that motivates the use of distributed storage codes in systems like Windows Azure and Hadoop~\cite{azure12,sathiamoorthy2013xoring}.

The central question is how to encode information so that a scaling number $k$ of read requests can be performed in parallel while the rate remains constant (and ideally approaches one). This problem is equivalent
to designing {\em multiset} batch codes with good parameters.
Multiset batch codes were also introduced in~\cite{batch}. 
These are batch codes that allow read requests that form a multiset, that is the same item $x_i$ from $\xv$ may be requested multiple times.
For any $k$ (not necessarily distinct) requests $x_{i_1}, \ldots , x_{i_k}$
there must be a partition of the $m$ servers into $k$ disjoint groups
$S_1, \ldots , S_k$, such that the $j$-th request $x_{i_j}$ can be 
recovered by reading at most $1$ symbol from each server of the group $S_j$.
Thus, in the context of distributed storage, multiset batch codes also allow multiple parallel reads, a very useful property when storing frequently accessed data. 
For example, if item $x_1$ is requested by $k$ different users, it can be retrieved simultaneously, by accessing distinct groups of servers, and 
reading only one symbol 
from each server.

The central code design question
for this scenario  involves minimizing the number of used servers $m$ 
while maximizing the rate $\rho=n/N$ for some given number of objects $n$ and required number of parallel reads $k$. 
Clearly, for any scheme $m \geq k$. Note that $m=k$ is achievable 
by $k$-replication but with vanishing rate $\rho=1/k$.
An important contribution of~\cite{batch} was the design of codes that for any constant rate $\rho<1$ achieve polynomial dependency of the number of servers
on $k$, that is  $m=k^c$ for some constant $c$ that depends on the rate $\rho$.
One limitation is that this exponent $c$ is fairly large and grows arbitrarily large as the rate $\rho$ 
approaches $1$.
We note that for low rates~\cite{batch} proposed a different construction that is near-optimal. 
Specifically, they showed that using Reed-Muller codes it is possible to 
obtain batch codes of rate $\rho=1/k^\epsilon$ and $m=k \log^{2+1/\epsilon+o(1) } k$ for any constant $\epsilon>0$.
It  remained open  if it is possible to construct multiset batch codes with 
rates approaching optimality $\rho=1-o(1)$. 
Further, it would be ideal to achieve this using few servers, {\em i.e.}, $m=k^c$ for some small exponent $c$.
By a connection between multiset batch codes and smooth codes
observed by Ishai et al. (Theorem 3.9 in \cite{batch}),
the recent constructions of high rate locally decodable codes
by Kopparty et al. \cite{KSY} imply multiset batch codes with these
properties, e.g. rate $1 - O(\log \log k / \log k)$ achieved with 
$m < k^{1+O(1/ \log \log k)}$, however this has not been explicitly 
mentioned in the literature. Other constructions of high rate locally decodable codes appear in
\cite{GKS12, HOW13}. 

As far as we know, 
the only constructions of batch codes
published after~\cite{batch} have been replication based batch codes
- referred to as \emph{combinatorial batch codes} in the literature 
(see e.g. \cite{BB14, SG14, RR08, PSW08, bhattacharya2012combinatorial, brualdi2010combinatorial, bujtas2011optimal}). Note that combinatorial batch codes 
do not provide multiset batch codes. 

A limitation of most previous multiset batch codes including 
all multiset batch codes of \cite{batch} with polynomial number of servers
$m = \text{poly}(k)$ is that they do not have explicit deterministic decoders.
%
They provide randomized decoding algorithms, and imply the existence
of the disjoint groups of servers required for decoding, but do not construct
the disjoint groups, and do not identify which servers to use for 
the parallel reads explicitly. 
The multiset batch codes with deterministic decoders in \cite{batch}
have superpolynomial number of servers with respect to $k$.

\noindent \textbf{Contributions:} 
We introduce the first families of multiset batch codes 
with explicit and deterministic decoding algorithms
and a polynomial number of servers with respect to $k$.
Our codes also 
achieve asymptotically optimal rates $1-o(1)$,
and asymptotically optimal fault tolerance. 
Our constructions exhibit a tradeoff between $m$ (the number of servers) and 
the rate as functions of $k$.
Note that we are trying to minimize $m$ and maximize the rate
as functions of $k$. 
We obtain several constructions that give different tradeoffs.
The smallest $m$ we can achieve
(as a function of $k$) while still having
rate $1 - o(1)$ is  
$m=k^{3+ \epsilon}$  for any fixed $\epsilon > 0$,
achieved with rate at least $1 - \frac{1}{k^{\epsilon /4}}$.
These tradeoffs are close to optimal for bipartite graph based 
multiset batch codes.
A different setting of parameters in this construction gives
$m = O(k^3)$ servers and rate at least  $1 - \frac{1}{c}$ for any 
constant $c \geq 2$.  
In another  construction we require $m= k^{4 +\epsilon}$  and achieve 
rate at least $1 - \frac{1}{k^{\epsilon}}$,
for any constant $1 \geq \epsilon > 0$.
The first construction gives a better tradeoff between the number of servers
and the rate when $\epsilon$ is small.
We also present  a construction that gives 
$m=k^4$ achieved with rate at least $1 - \frac{1}{\sqrt{k}}$. 
This gives a better rate than any of the other two constructions
when the number of servers $m = k^4$.
We also show that the zig-zag code 
construction~\cite{zigzag13} allows us to design batch codes with $m=\Theta(k^{k^{1+\epsilon} + 1 + \epsilon} )$ and rate at least $1-\frac{1}{k^{\epsilon}}$ for any $\epsilon > 0$. 
While they require substantially more servers, 
 zig-zag based codes offer the benefit of logarithmic locality, \textit{i.e.}, each of the parallel reads needs to contact only $O(\text{polylog} n)$ servers in the worst case. Table~\ref{table:parameters} summarizes the parameters obtained by various code constructions considered in this paper.

\begin{table*}[h!]
\begin{center}
    \begin{tabular}{  | p{3.8cm} | l | l |}
     \hline
    Construction & Rate & Number of servers ($m$) \\  \hline 
    Theorem~\ref{thm:balbuena2}, Section~\ref{sec:Balbuena_const} & $ \geq 1 - \frac{1}{k^{\epsilon /4}}$ (for $\epsilon >0$) &  $k^{3+ \epsilon}$ \\[3pt]\hline
    Theorem~\ref{thm:lazebnik}, Section~\ref{sec:Lazebnik_const} & $\geq 1 - \frac{1}{k^{\epsilon}}$ (for $1\geq \epsilon >0$) & $k^{4 + \tilde{\epsilon}} + o( k^{4 + \tilde{\epsilon}})$ \\[3pt] \hline
    Theorem~\ref{thm:decaen}, Section~\ref{sec:deCaen_const} &  $\geq 1 - 1/\sqrt{k}$ & $(1 + o_k(1)) k^4$  \\[3pt] \hline
    Theorem~\ref{thm:zigzag}, Section~\ref{sec:zigzag}   & $ \geq 1 - \frac{1}{k^{\epsilon}}$ (for $\epsilon >0 $)& $\Theta(k^{k^{1 + \epsilon} + 1 + \epsilon})$ \\[3pt] \hline
    Theorem~\ref{thm:balbuena3}, Section~\ref{sec:Balbuena_const}  & $ \geq 1 - \frac{1}{c}$ (for constant $c \geq 2)$ &  $O(k^{3})$ \\[3pt]\hline
    \end{tabular}
     \caption{Summary of the constructions of multiset batch codes presented in this paper. For the multiset batch codes constructed in Section~\ref{sec:Lazebnik_const}, $\tilde{\epsilon} < \epsilon $ is a rational number which depends on $\epsilon$. Note that we have deterministic decoding algorithm for all our constructions.}
      \label{table:parameters}
\end{center}
\end{table*}

Our constructions offer several benefits compared to prior work. 
The first is in decoding: the choice of which servers to use for each of 
the $k$ parallel reads is deterministic and obtained using an efficient 
algorithm. This is in contrast with previous constructions of multiset batch 
codes with polynomial number of servers with respect to $k$ 
that choose these sets randomly and 
only show the existence of assignments of disjoint groups of servers to
read requests.
The second benefit is that our multiset batch codes achieve rate
of the form $ 1 - \frac{1}{\text{poly}(k)}$ with a polynomial number 
of servers, which is larger than the rate $1 - O(\log \log k / \log k)$ 
of batch codes
implied by the recent high rate locally decodable codes of \cite{KSY}. 
Finally, our batch codes offer fault tolerance. Specifically, we show how to augment the code with extra parity symbols to provide distance asymptotically approaching that of Maximum Distance Separable (MDS) codes.
We note that while this transformation can be applied to any batch code,
batch codes with rate $1-o(1)$ are necessary to obtain 
asymptotically optimal fault tolerance.

\noindent \textbf{Techniques:} Our main technical innovation is a graph-theoretic method of designing multiset batch codes using dense bipartite graphs with no short cycles.
Consider first the case where each server stores one symbol (\textit{i.e.} $N=m$), called a \textit{primitive}
batch code. Primitive batch codes are subsequently used as building blocks to construct more general batch codes.

Consider a bipartite graph with $n$ left nodes and $N-n$ right nodes. Each left node corresponds to a data symbol and each 
right node to a parity symbol. Let the left nodes all have degree $k$, indicating that each data symbol $x_i$ is included in $k$ parities. It is not hard to see that if this bipartite graph has no 4-cycle, these $k$ parities must have disjoint neighborhoods, 
excluding $x_i$. Therefore, if there are no 4-cycles, each symbol $x_i$ can be reconstructed $k$ times in parallel, by reading these parities and their disjoint neighborhoods. 
  
This simple observation allows $k$ parallel reads for one symbol but cannot support for example $k/2$ reads for $x_1$ and 
$k/2$ reads for $x_2$ since all the parities of $x_1$ and $x_2$ might involve other common data symbols. Our main lemma shows that if there are no 4-cycles or 6-cycles, the number of parities that are blocked by each read request can be bounded by a local argument and any $k$ reads can be supported. 
Using this graph theoretic lemma we show that bipartite graphs with $n$ left nodes, $N-n$ right nodes, left degree at least $k$ and girth at least $8$
yield primitive multiset batch codes with parameters $(n,N,k,m)$. 
Therefore optimizing batch code parameters becomes an extremal graph problem: 
Maximize the left minimum degree in a (simple, undirected) bipartite graph with $N$ vertices so that it contains no 4-cycle or 6-cycle subgraphs. Here subgraphs are not necessarily induced and this condition is equivalent to the bipartite graph having girth (length of shortest cycle) at least $8$. 

The difficulty in code constructions is now clear: making the left degree $k$ large without short cycles requires many right nodes. However, as 
the number of right nodes (\textit{i.e.} parities) is $N-n$, this pushes the rate $n/N$ low. The question of constructing dense bipartite graphs of high girth and in particular bipartite graphs without 4- and 6-cycles and
as many edges as possible has been extensively studied. 
The strongest constructions of such graphs have been given by 
Lazebnik, Ustimenko, Woldar \cite{lazebnik} and by de Caen 
and Sz\'ekely \cite{decaen}. The construction by Lazebnik \textit{et al.}
is obtained by embedding Chevalley group geometries in Lie algebras~\cite{ustimenko}. The constructions of de Caen and Sz\'ekely are based on known constructions
of generalized quadrangles \cite{quadbook1, quadbook2}.
We modify a construction of Balbuena \cite{balbuena2009}
based on Latin Squares to obtain our results with $m=k^{3+ \epsilon}$
and rate at least $1 - \frac{1}{k^{\epsilon /4}}$. This 
gives  close to optimal tradeoffs  for bipartite graph based 
multiset batch codes, since as we discuss in Section~\ref{limit}
rate at least  $1 - \frac{1}{k^{\epsilon /4}}$ implies $m = \Omega(k^{3+ \epsilon/2})$.
Note that the graphs of Balbuena \cite{balbuena2009} only give rate $1/2$, 
but we apply a simple transformation to amplify the rate.
We could apply a similar transformation to any of our bipartite graph 
based constructions, but these would not lead to significantly different 
tradeoffs.

\subsection{Related work}

Our work is inspired by work of Ishai et al.~\cite{batch}. Batch codes were originally motivated for load balancing applications which are traditionally different from erasure codes like Reed-Solomon~\cite{reed1960polynomial} and information dispersal~\cite{rabin1989efficient} used for distributed data storage~\cite{plank1997tutorial}. 
The difference is that batch codes are useful for storing many small items while standard erasure codes provide fault tolerance for a single large object split in multiple symbols. 

One issue that blurs this difference, however, is the so-called repair problem of distributed storage codes~\cite{dimakis,dimakis2011survey}. When a systematic code is used for storage across servers a very common operation is the access of a single systematic symbol.
Typically, each symbol of a code is stored in a different server. When the particular server storing the desired systematic symbol has failed or is temporarily unavailable the problem of reconstructing a single symbol from as few other codeword symbols as possible arises. 
The locality of a symbol~\cite{Gopalan12,PapDim12} is the smallest number of other symbols that need to be read to reconstruct it. We refer to that set of symbols as a repair group for the given symbol. Batch codes enable multiple ($k$) repair groups for each symbol that are also disjoint and hence allow these reads to be performed in parallel, a property very useful for high performance storage systems. Since these groups must be disjoint, some of them must be small so batch codes implicitly enforce locality. The weaker problem of enabling multiple parallel reads of a single symbol was recently investigated in~\cite{availability}. The constructions we present in this paper resolve an open question stated in~\cite{availability}. 

Fault tolerance is another useful property for distributed storage codes. 
Gopalan~\textit{et al.}~\cite{Gopalan12} determined the largest possible distance for codes with locality, generalizing the Singleton bound~\cite{MacSlo}.
This was recently further generalized for nonlinear codes~\cite{PapDim12} and several explicit constructions of locally repairable codes have been investigated~(\textit{e.g.} \cite{TamoPapDim13,oggier_proj,RKSV12,KPLK12,TamoBarg}). Note that locally repairable codes are related but fundamentally different from locally decodable codes~\cite{yekhanin2010locally}. Prior batch code constructions with good rates have very poor distance properties, namely $d=o(n)$. 
We construct  batch codes that have near-optimal distance, linear in $n$ and asymptotically approaching the Singleton bound. 

We note that girth and graph representations have been considered before
(see e.g. \cite{PSW08, BB14} in the context of $2$-uniform combinatorial
batch codes, that is replication based batch codes where each message symbol is
repeated exactly twice. However, this connection is very different from our 
technical contributions. $2$-uniform combinatorial batch codes can be
viewed as graphs the following way: since each symbol is stored at 
exactly two  servers, one can represent these codes by a graph whose
vertex set is the set of servers, and an edge is placed for each symbol, 
connecting the two servers where the symbol is stored.
Note that combinatorial batch codes repeating symbols more than twice
give rise to hypergraphs. In the case of 2-uniform combinatorial  batch codes,
the corresponding graph cannot contain certain configurations as 
subgraphs, and in particular the girth of the graph is also relevant.
\cite{BB14} uses a construction of Lazebnik, Ustimenko, Woldar 
\cite{lazebnik2} of high girth graphs to construct 2-uniform combinatorial
batch codes. However, this is a different construction and a different paper
than the construction in \cite{lazebnik} that we use, and
our graph representation and connection between  high girth graphs and 
multiset batch codes is completely different.


%% file: organization.tex
\subsection{Organization of the paper}
\label{sec:organization}
In Section~\ref{sec:prelim}, we provide necessary background on batch codes, including the gadget lemma~\cite{batch}. 
The gadget lemma allows us to convert a primitive batch code to a 
general batch code  without the restriction that the  number of servers 
be equal to the number of code symbols. 
We describe our main technical  contribution connecting dense high-girth graphs
to multiset batch codes  in Section~\ref{sec:main_sec}. 
We first associate a natural bipartite graph representation $\Gc_{\Cc}$ 
with a systematic linear code ${\Cc}$ in Section~\ref{sec:graph}, 
and define the notion of repair groups for such codes in 
Section~\ref{sec:repair_groups}. 
In Section~\ref{sec:main_result}, we establish that 
having an induced subgraph of $\Gc_{\Cc}$ with left degree at least $k$ and girth at least $8$ is a sufficient condition for the code to be a primitive multiset batch code, which can support any sequence of $k$ reads requests. 
In Section~\ref{sec:constructions} we use this result with various 
explicit constructions of dense bipartite graphs with girth at least $8$ 
to obtain multiset batch codes with rate $1 - o(1)$.
In Section~\ref{limit}  we discuss the limitations of our bipartite graph 
based approach. 
In Section~\ref{sec:distance}, we combine the load balancing aspect of batch codes with fault tolerance, and describe a simple way to obtain batch codes that are arbitrarily close to the Singleton bound.

\textbf{Notation:} A quick note on the notation that we follow throughout the paper. We use boldface lower case letters to denote vectors. 
For a strictily positive integer $n$, $[n]$ denotes the set $\{1, 2,\ldots, n\}$. For two integers $0 \leq n_1 \leq n_2$, we use $[n_1:n_2]$ to represent the set $\{n_1, n_1 + 1,\ldots, n_2\}$. The notations which are specific to various sections of the paper are defined where we use them.

%% file: background.tex
\section{Background}
\label{sec:prelim}

Let $\Sigma$ denote a finite alphabet\footnote{We subsequently assume $\Sigma$ to be a finite field.} and $\xv = (x_1, x_2,\ldots, x_n) \in \Sigma^{n}$ represent a length-$n$ message vector. We use the (multi-)set $\{i_1, i_2,\ldots, i_k\}$ to represent the sequence of $k$ read requests that correspond to the message symbols $x_{i_1}, x_{i_2},\ldots, x_{i_k}$. Note that a multiset with $k$ elements from $[n]$ denotes a sequence of $k$ read requests where some message symbols may be requested multiple times. In what follows, we also refer to a sequence of $k$ read requests as a {\em $k$-request pattern}. 
We start with the definition of batch codes~\cite{batch}:

\begin{definition}{(Batch codes \cite{batch})}
An $(n, N, k, m, t)$ batch code encodes $n$ message symbols over $\Sigma$ to $N$ code symbols over $\Sigma$. These $N$ code symbols are partitioned into $m$ buckets (servers) such that any $k$-request pattern represented by a set with $k$ distinct elements from $[n]$, {\em i.e.}, parallel reads for $k$ distinct message symbols, can be served by downloading at most $t$ code symbols from each of the $m$ buckets. 
\end{definition}

The notion of batch codes is further strengthened into {\em multiset batch codes} which also support the $k$-request patterns represented by mutlisets with $k$ elements from the set $[n]$:

\begin{definition}{(Multiset batch codes \cite{batch})}
An $(n, N, k, m, t)$ {multiset} batch code encodes $n$ message symbols over $\Sigma$ to $N$ code symbols over $\Sigma$. These $N$ coded symbols are partitioned into $m$ buckets such that any $k$-request pattern represented by a multiset with $k$ elements from $[n]$ can be served  by downloading at most $t$ code symbols from each of the $m$ buckets. In addition, all read requests are required to be served from the code symbols downloaded from disjoint sets of buckets. 
\end{definition}

\begin{remark}
A particular sub-class of batch codes where the number of buckets is equal to the length of codewords, {\em i.e.}, $m = N$, is referred to as {\em primitive batch codes}. Note that primitive batch codes correspond to batch codes with each of their buckets comprising exactly one code symbol. 
\end{remark}

\begin{remark}
In this paper we focus on batch codes with $t = 1$, {\em i.e.}, at most only one symbol can be read from each of the buckets during the process of serving a $k$-request pattern. Given an $(n, N, k, m, 1)$ batch code,  one can obtain an $(n, N, tk, m, t)$ batch code and an $(n, N, k, \ceilb{\frac{m}{t}}, t)$ batch code~\cite[Lemma~2.4.2]{batch}. In what follows, we refer to an  $(n, N, k, m, 1)$ batch code as an $(n, N, k, m)$ batch code.
\end{remark}

In \cite{batch}, Ishai {\em et al.} highlight a basic operation to convert a primitive batch code to a general batch code with suitable parameters. The following result describes this operation.
\begin{lemma}{(Gadget lemma~\cite{batch})}
\label{lem:gadget}
Given an $(n, N, k, m)$ multiset batch code $\Cc$, for any $g \in \NN$, it is possible to construct a $(gn, gN, k, m)$ multiset batch code $g \cdot\Cc$.
\end{lemma}

For completeness, we include a proof of the gadget lemma. 

\begin{proof}~
We prove the Lemma by presenting an explicit construction for a $(gn, gN, k, m)$ multiset batch code $g\cdot\Cc$ which utilizes the given $(n, N, k, m )$ multiset batch code $\Cc$. Let $\xv = (x_1, x_2,\ldots, x_{gn})$ be a message vector of length $gn$. We partition the message vector into $g$ sub-vectors $\xv^1, \xv^2,\ldots, \xv^{g}$ such that 
$$
\xv^{i} = (x_{(i-1)n + 1}, x_{(i-1)n + 2},\ldots, x_{in}),~~\text{for}~i\in [g].
$$
Assuming that $\Cc(\xv^{i})$ denotes the length-$n$ codeword  corresponding to $\xv^{i}$ in $\Cc$, we assign the following codeword in $g\cdot\Cc$ to the message vector $\xv$.
$$
g\cdot\Cc(\xv) = \cv = ({\cv}_1, {\cv}_2,\ldots, \cv_N).
$$
Here, for $j \in [N] = [m]$, $\cv_j = (\Cc(\xv^{1})_j, \Cc(\xv^{2})_j,\ldots, \Cc(\xv^{g})_j)$ represents the $j$-th bucket associated with the codeword $\cv \in g\cdot\Cc(\xv)$. We now establish that $g\cdot\Cc$ is indeed a  $(gn, gN, k, m)$ multiset batch code.

It is clear from the construction described above that the codewords in $g\cdot\Cc$ are obtained by stacking up $g$ codewords in $\Cc$, one for each of the $g$ sub-vectors $\{\xv^i\}_{i \in [g]}$. Given this structure of the codewords in $g\cdot\Cc$, we can associate any $k$-request pattern for $\xv$ to a $k$-request pattern for $\xv^1 = (x_1, x_2,\ldots, x_n)$: for any $j \in [n]$, we map a read request for the message symbol $x_{(i - 1)n + j}$ to a read request for the message symbol $x_{j}$. For $j \in [n]$, let $k_j$ denote the number of read requests from the $k$-request pattern for $\xv$ that are mapped to the read requests for the message symbol $x_j$. Note that we have $\sum_{j = 1}^nk_j = k$. Since $\Cc$ is an $(n, N, k, m)$ primitive batch code, we can assign the $k$ read requests for the symbols in the sub-vector $\xv^1$  to disjoint sets of buckets. Now, for any $j \in [n]$, one can serve all $k_j$ read requests mapped to $x_j$ by reading the appropriate code symbols from the assigned sets of buckets.
\end{proof}

%% file: main_results.tex
\section{Obtaining Multiset Batch Codes from Bipartite Graphs}
\label{sec:main_sec}

This section presents the primary technical contribution of this paper. 
Here, we focus on systematic linear codes and provide a sufficient condition for such a code to be a primitive multiset batch code. For the remainder of the paper,  we assume $\Sigma$ to be a finite field over which our linear codes are defined. In particular, the parity symbols of our codes correspond to linear combinations of the message symbols with coefficients chosen from $\Sigma$.

For every systematic linear code, we associate a bipartite graph with the code. We then show that an $[n, N]$ systematic code\footnote{An $[n, N]$ systematic code is a linear code over finite field $\Sigma$ that encodes $n$ message symbols over $\Sigma$ to length-$N$ codewords over $\Sigma$. Moreover, without loss of generality we assume that the $n$ message symbols (systematic part) appear at the first $n$ positions of the codewords in $\Cc$.} is an $(n, N, k, m = N)$ primitive multiset batch code if the associated bipartite graph has an induced subgraph with girth at least $8$ and left degree at least $k$. The proof of the fact that a systematic code which satisfies the aforementioned graph theoretic conditions is a multiset batch code has an added advantage. The proof involves showing that for such codes any $k$-request pattern can be assigned to disjoint sets of code symbols by a greedy and efficient algorithm (see Figure~\ref{algo:batch}).

\subsection{Graph representation for systematic codes}
\label{sec:graph}

Given an $[n, N]$ systematic code $\Cc$, we define a bipartite graph $\Gc_{\Cc} = (\Vc_{1} \cup \Vc_{2}, \Ec)$ with $|\Vc_1| = n$ and $|\Vc_2| = N - n$. The $n$ left nodes and the $N - n$ right nodes in the graph $\Gc_{\Cc}$ respresent the $n$ message symbols and the $N- n$ parity symbols of the code $\Cc$, respectively. Moreover, there exists an edge $(i, j) \in \Ec \subseteq \Vc_1 \times \Vc_2$ iff the $i$-th message symbol participates in the $j$-th parity symbol. Alternatively, given $G = [I_{n \times n}~|~E_{n \times (N-n)}]$ as the $n \times N$ generator matrix of the code $\Cc$, the matrix $E$ is supported on the incidence matrix of the bipartite graph $\Gc_{\Cc}$. 

\begin{example}
\label{ex:graph}
Let $\Cc$ be a $[3, 8]$ systematic code which encodes the message vector $\xv = (x_1, x_2, x_3)$ to a length-$8$ codeword $\cv$ such that 
$$
\cv = (c_1, c_2,\ldots, c_8) = (x_1, x_2, x_3, x_1 + x_2, x_1, x_2 + x_3, x_1 + x_3, x_1 + x_2 + x_3).
$$
Figure~\ref{fig:graph_C} describes the graph representation for $\Cc$. 
\end{example}

\begin{remark}
Note that the graph representation considered here is different from the prevalent Tanner graph representation of the parity check matrix of a code~\cite{modern_coding}. In the graph $\Gc_{\Cc}$, all nodes represent the symbols of a generic codeword with the left nodes and the right nodes representing the message and the parity symbols, respectively. On the other hand, the {\em variable} and {\em check} nodes in a Tanner graph represent the code symbols  and the different constraints satisfied by the codewords, respectively. However, there is a correspondence between these two graph representations for a systematic code as one can be obtained from the other.
\end{remark}

\begin{figure}[htbp]
	\centering
		\includegraphics[width=0.40\textwidth]{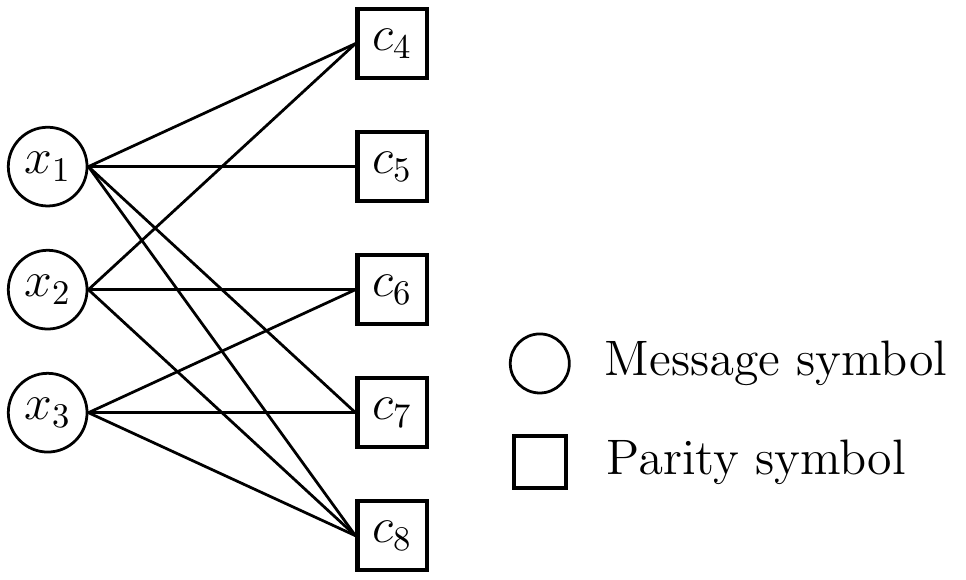}
          \caption{Graph representation $\Gc_{\Cc}$ for the $[3, 8]$ code described in Example~\ref{ex:graph}.}
	\label{fig:graph_C}
\end{figure}

\subsection{Repair groups of a systematic code}
\label{sec:repair_groups}
We now define the notion of repair groups for a systematic linear code $\Cc$. For $i \in [n]$, let $c_{i_1}, c_{i_2},\ldots, c_{i_{u(i)}}$ denote the $u(i)$ parity symbols that involve the message symbol $x_i$.  For $j \in \{i_1, i_2,\ldots, i_{u(i)}\}$, we use $\xv_{\Gamma(i, j)} \subseteq \{x_1,\ldots, x_{i-1}, x_{i + 1},\ldots, x_n\}$ to represent the collection of the message symbols other than $x_i$ that participate in the construction of the parity symbol $c_{j}$. This implies that for each $j \in \{i_1, i_2,\ldots, i_{u(i)}\}$, one can obtain $x_i$ by removing the contributions of the symbols in $\xv_{\Gamma(i, j)}$ from $c_{j}$. Noting this, we refer the $u(i)$ collections of the symbols $\big\{\xv_{\Gamma(i, j)} \cup \{c_{j}\}\big\}_{j \in \{i_1, i_2,\ldots, i_{u(i)}\}}$ as $u(i)$ repair groups for the message symbol $x_i$.

In terms of the graph representation $\Gc_{\Cc}$, for $i \in [n]$, the left node $x_i$ has degree $u(i)$. The right nodes associated with the parity symbols $c_{i_1}, c_{i_2},\ldots, c_{i_{u(i)}}$ constitute the $u(i)$ neighbors of the left node $x_i$. Moreover, for $j \in \{i_1, i_2,\ldots, i_{u(i)}\}$, the left nodes associated with the message symbols $\xv_{\Gamma(i, j)}$ correspond to all the neighbors of the right node $c_{j}$ other than the left node $x_i$.

For the $[3, 8]$ code described in Example~\ref{ex:graph} (see Figure~\ref{fig:graph_C}), the message symbol $x_1$ participates in four parity symbols $c_4, c_5, c_7,  c_8$. Therefore, the message symbol $x_1$ has the following four repair groups, 
\begin{align*}
\{x_2, c_4\},~\{c_5\},~\{x_3, c_7\},~\{x_2, x_3, c_8\}.
\end{align*}
The repair groups for the message symbols $x_2$ and $x_3$ can be listed in a similar manner.

\subsection{Sufficient condition for multiset batch codes}
\label{sec:main_result}

As our main technical contribution, the following result  presents a sufficient condition for a linear  systematic code to be a primitive multiset batch code.

\begin{theorem}
\label{thm:no_6cycle}
Let $\Cc$ be an $[n, N]$ systematic code with the graph representation $\Gc_{\Cc} = (\Vc_{1}, \Vc_{2}, \Ec)$. Assume that there exists an induced subgraph $\widetilde{\Gc}_{\Cc} =  (\Vc_{1}, \widetilde{\Vc}_{2} \subseteq \Vc_{2}, \widetilde{\Ec})\footnote{For $\widetilde{\Vc}_{2} \subseteq \Vc_2$, the edge set associated with $\widetilde{\Gc}_{\Cc}$ is defined as $ \widetilde{\Ec} := \{(i, j) \in \Ec:~j \in \widetilde{\Vc}_{2}\}.$} \subseteq \Gc_{\Cc}$ such that the following two conditions hold.
\begin{itemize}
\item[(i)] Each node in $\Vc_{1}$ has degree at least $k$ in the bipartite graph $\widetilde{\Gc}_{\Cc}$.
\item[(ii)] The bipartite graph $\widetilde{\Gc}_{\Cc} \subseteq \Gc_{\Cc}$ has girth (length of the shortest cycle) at least $8$.
\end{itemize}
Then, $\Cc$ is an $(n, N, k, m = N)$ primitive multiset batch code.
\end{theorem}

Before proceeding to establish Theorem~\ref{thm:no_6cycle}, we present Lemma~\ref{lem:no_4cycle} and \ref{lem:no_6cycle}. These two lemmas characterize the interaction among those repair groups of the message symbols that are associated with a suitable induced subgraph $\widetilde{\Gc}_{\Cc} \subseteq \Gc_{\Cc}$. Note the requirement (i) in Theorem~\ref{thm:no_6cycle} requires the graph to be dense, while the requirement (ii) demands that the graph be not too dense. 

\begin{lemma}
\label{lem:no_4cycle}
Let $\Cc$ be an $[n, N]$ systematic code with the graph representation $\Gc_{\Cc} = (\Vc_{1}, \Vc_{2}, \Ec)$. Assume that there exists an induced subgraph $\widetilde{\Gc}_{\Cc} = (\Vc_{1}, \widetilde{\Vc}_{2} \subseteq \Vc_{2}, \widetilde{\Ec}) \subseteq \Gc_{\Cc}$ such that the following two conditions hold.
\begin{itemize}
\item[(i)] Each node in $\Vc_{1}$ has degree at least $k$ in the bipartite graph $\widetilde{\Gc}_{\Cc}$,
\item[(ii)] The bipartite graph $\widetilde{\Gc}_{\Cc} \subseteq \Gc_{\Cc}$ has no $4$-cycle, {\em i.e.}, cycle of length $4$.
\end{itemize}
Then, each message symbol has at least $k$ disjoint repair groups.
\end{lemma}

\begin{proof}~
Let $\widetilde{\Cc}$ denote the systematic code associated with the subgraph $\widetilde{\Gc}_{\Cc} \subseteq \Gc_{\Cc}$. Without loss of generality, we assume that $\widetilde{\Vc}_{2} \subseteq \Vc_{2}$ corresponds to the first $|\widetilde{\Vc}_{2}|$ out of $|\Vc_{2}| = N - n$ right nodes of the graph $\Gc_{\Cc}$\footnote{If this is not the case, we can re-index the parity symbols of the codewords in $\Cc$ to get this property.}. Then, every codeword $\cv = (c_1, c_2,\ldots, c_{N}) \in \Cc$ is associated to a unique codeword $\widetilde{\cv} = (c_1, c_2,\ldots, c_{n + |\widetilde{\Vc}_{2}|}) \in \widetilde{\Cc}$ obtained by removing the last $|\Vc_2| - |\widetilde{\Vc}_2|$ code symbols from the codeword $\cv$. In the rest of the proof we focus only on the repair groups that contain the parity symbols associated with the set $\widetilde{\Vc}_{2}$. In other words, we work with the code $\widetilde{\Cc}$.

\begin{figure}[htbp]
	\centering
		\includegraphics[width=0.55\textwidth]{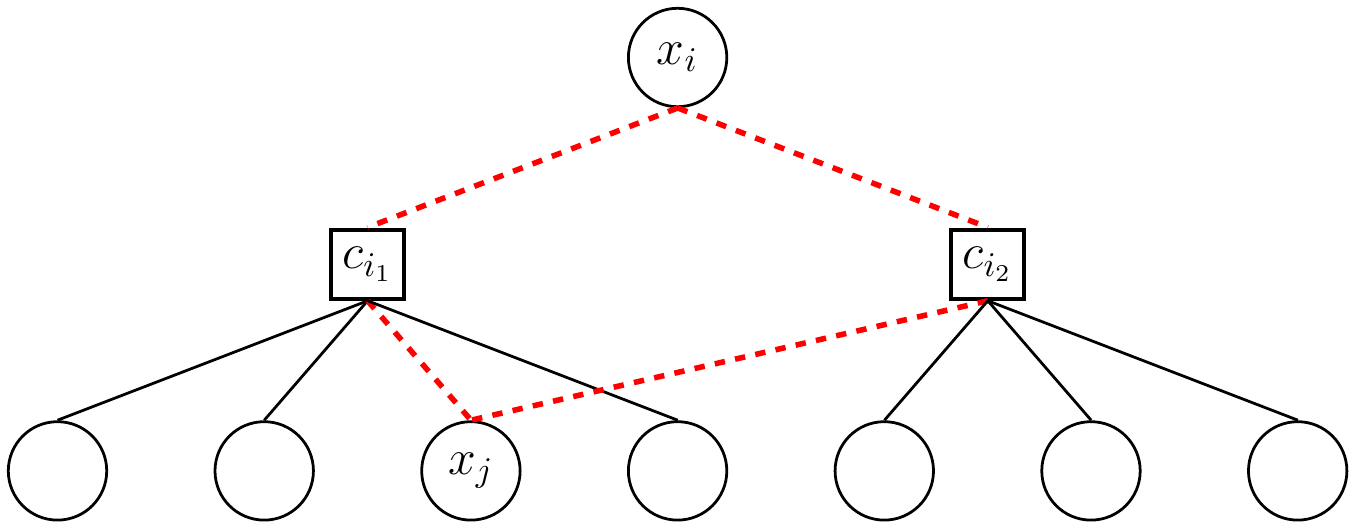}
          \caption{Illustration of a $4$-cycle $\{x_i, c_{i_1}, x_j, c_{i_2}\}$ when two repair groups of the message symbol $x_{i}$ have a common message symbol $x_j$.}
	\label{fig:4_cycle}
\end{figure}

Since each left node has degree at least $k$ in the bipartite graph $\widetilde{\Gc}_{\Cc}$, this implies that for any $i \in [n]$, the message symbol $x_i$ has at least $k$ repair groups (see Section~\ref{sec:repair_groups}). Next, we argue that all of these repair groups are disjoint, {\em i.e.}, they do not  have any common code symbol. In order to establish contradiction, we assume that for an $i \in [n]$, there exists two repair groups of the message symbol $x_i$ with at least one common code symbol. Since two repair groups for a message symbol involve distinct parity symbols, let $c_{i_1}$ and $c_{i_2}$ denote the parity symbols in these two intersecting repair groups. Assume that $x_j$, for $j \in [n]\backslash \{i \}$, be one of the common message symbol in the two underlying repair groups. This implies that the nodes $\{x_i, c_{i_1}, x_j, c_{i_2}\}$ form a cycle in the graph $\widetilde{\Gc}_{\Cc}$ (see Figure~\ref{fig:4_cycle}). This leads to contradiction as the graph $\widetilde{\Gc}_{\Cc}$ has no $4$-cycle.
\end{proof}

%
%

\begin{lemma}
\label{lem:no_6cycle}
Let $\Cc$ be an $[n, N]$ systematic code with the graph representation $\Gc_{\Cc} = (\Vc_{1}, \Vc_{2}, \Ec)$. Assume that there exists an induced subgraph $\widetilde{\Gc}_{\Cc} =  (\Vc_{1}, \widetilde{\Vc}_{2} \subseteq \Vc_{2}, \widetilde{\Ec})  \subseteq \Gc_{\Cc}$ such that the following two conditions hold.
\begin{itemize}
\item[(i)] Each node in $\Vc_{1}$ has degree at least $k$ in the bipartite graph $\widetilde{\Gc}_{\Cc}$,
\item[(ii)] The bipartite graph $\widetilde{\Gc}_{\Cc} \subseteq \Gc_{\Cc}$ has girth (length of the shortest cycle) at least $8$.
\end{itemize}
Then,  each message symbol has at least $k$ disjoint repair groups. Moreover, for any $i, j \in [n]$ with $i \neq j$, any one of the disjoint repair groups for the message symbol $x_i$ has common symbols with at most one of the disjoint repair groups for the message symbol $x_j$.
\end{lemma}

\begin{proof}~
 Here, we use ideas which are similar to those employed in the proof of Lemma~\ref{lem:no_4cycle}. In particular, we associate a code $\widetilde{\Cc}$ to the subgraph $\tilde{\Gc}_{\Cc} \subseteq \Gc_{\Cc}$. Note that for every codeword $\cv = (c_1, c_2,\ldots, c_{N}) \in \Cc$, there exists a unique codeword $\widetilde{\cv} = (c_1, c_2,\ldots, c_{n + |\widetilde{\Vc}_{2}|}) \in \widetilde{\Cc}$. Again in the rest of the proof we only focus on  those repair groups that are associated with $\widetilde{\Cc}$, {\em i.e.}, the repair groups that do not contain the codes symbols $c_{n + |\widetilde{\Vc}_{2}|+1}, \ldots, c_{n}$  from a codeword $\cv = (c_1,\ldots, c_{N}) \in {\Cc}$.  

Since the sub graph $\widetilde{\Gc}_{\Cc}$ has girth at least $8$, it does not contain a $4$-cycle. It then follows from  Lemma~\ref{lem:no_4cycle} that each of the $n$ message symbols has at least $k$ disjoint repair groups. Next, we establish the second claim in the lemma that for any $i, j \in [n]$ with $i \neq j$, any one of the disjoint repair groups for the message symbol $x_i$ has common symbols with at most one of the disjoint repair groups for the message symbol $x_j$. 

\begin{figure*}[t!]
        \centering
        \begin{subfigure}[b]{0.48\textwidth}
	\centering
		\includegraphics[width=0.63\textwidth]{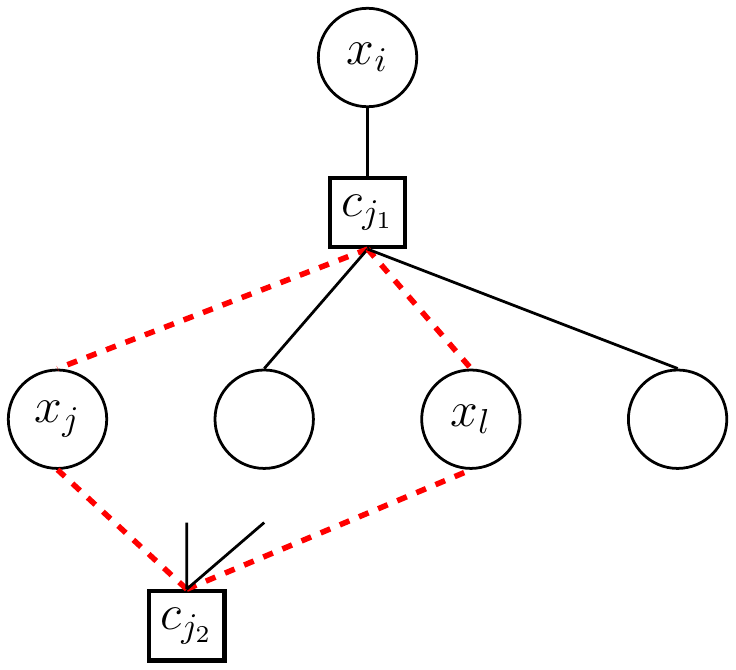}
                      \caption{$4$-cycle $\{x_j, c_{j_1}, x_l, c_{j_2}\}$}
                       \label{fig:6cycle_case1}
        \end{subfigure}%
       \quad
        \begin{subfigure}[b]{0.48\textwidth}
	\centering
		\includegraphics[width=0.63\textwidth]{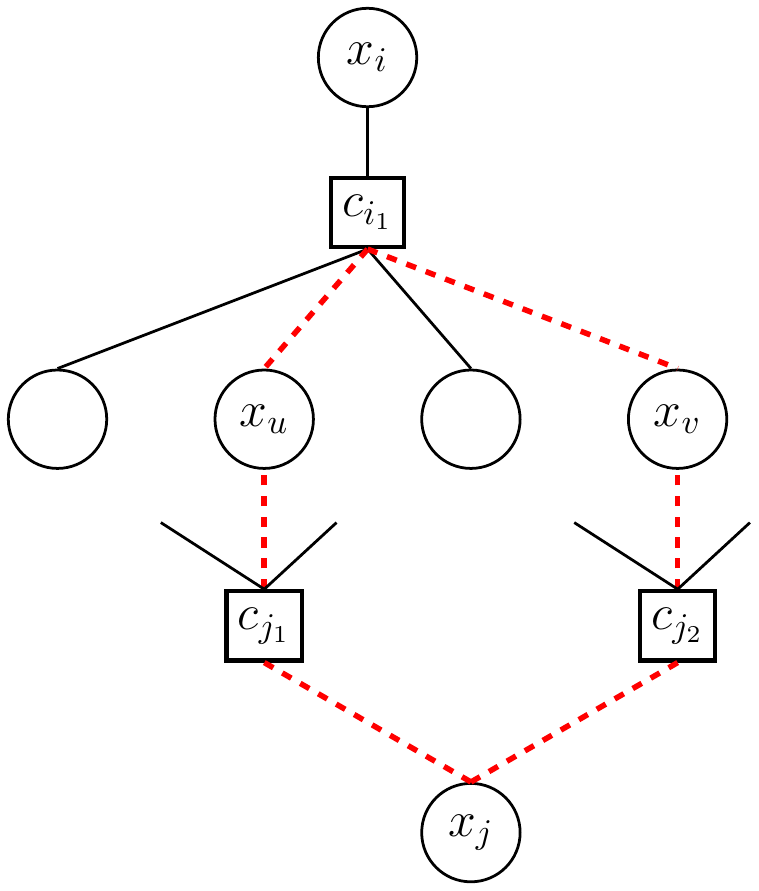}
                      \caption{ $6$-cycle $\{x_j, c_{j_1}, x_u, c_{i_1}, x_v, c_{j_2}\}$}
                      \label{fig:6cycle_case2}
        \end{subfigure}
        \caption{Illustration of the cycles in $\widetilde{\Gc}_{\Cc}$ when one repair group of the message symbol $x_i$ intersects with two disjoint repair groups of the message symbol $x_j$.}
\end{figure*}

Toward this, we assume that the opposite is true and there exist two message symbols, say $x_i$ and $x_j$, such that at least one repair group of the message symbol $x_i$  has common code symbols with at least two repair groups of the message symbol $x_j$. Let $c_{i_1}$ be the parity symbol corresponding to one such repair groups of $x_i$ that intersects with two repair groups of the message symbol $x_j$. Assume that  $c_{j_1}$ and $c_{j_2}$ represent the distinct parity symbols associated with the underlying two repair groups the message symbol $x_j$. There are two cases that need to be considered: 
\begin{itemize}
\item \textbf{$i_1 = j_1$ or $i_1 = j_2$:}~Without loss of generality, we assume that $i_1 = j_1$. Let $x_l$ be the message symbol that is common between the repair group of $x_i$ and the repair group of $x_j$ corresponding to the parity symbol $c_{j_2}$.  As described in Figure~\ref{fig:6cycle_case1}, this implies the existence of a $4$-cycle $\{x_j, c_{j_1}, x_l, c_{j_2}\}$ in the bipartite graph $\widetilde{\Gc}_{\Cc}$.
\item \textbf{$i_1 \neq j_1$ and $i_1 \neq j_2$:} In this case, let $x_{u}$ represent a common symbol between the repair group of $x_i$ and the {\em first repair group} of $x_j$ corresponding to the parity symbol $c_{j_1}$. Similarly, assume that $x_{v}$ denotes a symbol which is common between the repair group of $x_i$ and the {\em second repair group} of $x_j$ corresponding to the parity symbol $c_{j_2}$. This translates to a $6$-cycle $\{x_j, c_{j_1}, x_u, c_{i_1}, x_v, c_{j_2}\}$ in the bipartite graph $\widetilde{\Gc}_{\Cc}$ (see Figure~\ref{fig:6cycle_case2}).
\end{itemize}

Since the girth of the bipartite graph $\widetilde{\Gc}_{\Cc}$ is at least $8$, both cases above lead to contradiction. This completes the proof.

\end{proof}

Essentially, Lemma~\ref{lem:no_6cycle} implies that given an induced subgraph $\widetilde{\Gc}_{\Gc}$ with girth at least $8$, among the repair groups associated with $\widetilde{\Gc}_{\Cc}$, each repair group of one message symbol can block at most one repair group of any other message symbol. In the following, we utilize this to prove Theorem~\ref{thm:no_6cycle}.

\begin{proof}[Proof of Theorem~\ref{thm:no_6cycle}]
Given the code $\Cc$, we assign each symbol from the codeword $\cv \in \Cc$ to a different bucket. This implies that $m = N$. Here, we provide an algorithmic proof of the fact that $\Cc$ can support any $k$-request pattern by serving different read requests from the disjoint sets of buckets.  In particular, under the assumptions on the bipartite graph $\widetilde{\Gc}_{\Cc}$ stated in Theorem~\ref{thm:no_6cycle}, we provide an algorithm which always maps $k$ read requests in any $k$-request pattern to $k$ disjoint sets of code symbols (buckets). The algorithm described here is deterministic as opposed to the randomized algorithm used for Subset codes  by Ishai {\em et al.} in~\cite{batch}.

Let $\{i_1, i_2,\ldots, i_{k}\} \subset [n]$ be a multiset corresponding to a $k$-request pattern. For the $k$-request pattern and $j \in [n]$, let $k_j$ denote the number of read requests for the message symbol $x_j$ in the $k$-request pattern, {\em i.e.}, 
$$
k_j = |\{u \in [k]: i_u = j\}|,~~\text{for}~j \in [n].
$$
Note that $\sum_{j = 1}^{n}k_j = k$. Moreover, each $k$-request pattern can be alternatively defined by a set 
$$\big\{(1, k_1), (2, k_2),\ldots, (n, k_n)\big\} \in [n]\times[k],$$
 where $(j, k_j)$ implies that there are $k_j$ read requests for the message symbol $x_j$  in the $k$-request pattern. We now use the algorithm described in Figure~\ref{algo:batch} to map the $k$ read requests to disjoint sets of code symbols (buckets). 

It follows from Lemma~\ref{lem:no_6cycle} that each message symbol has at least $k$ disjoint repair groups in $\Cc$. Let $j^{\ast}  = \min\{j \in [n]: k_j > 0\}$. The algorithm can successfully assign the $k_{j^{\ast}}$ read requests for the message symbol $x_{j^{\ast}}$ to the buckets storing $x_{j^{\ast}}$ itself and $k_{j^{\ast}} - 1$ disjoint repair groups of $x_{j^{\ast}}$. In order to complete the proof of Theorem~\ref{thm:no_6cycle}, it is enough to show that for any $j > j^{\ast}$, the algorithm is always able to find $k_j$ sets of disjoint code symbols (buckets) to serve $k_j$ read requests for the message symbol $x_j$.

\begin{figure}[t!]
\algrule[1pt]
\textbf{Algorithm:} Assignment of $k$ read requests to $k$ disjoint sets of code symbols (buckets).
\algrule[1pt]
\begin{algorithmic}[1]
\REQUIRE An $[n, N]$ systematic code $\Cc$ that satisfies the assumptions stated in Theorem~\ref{thm:no_6cycle}, and $k$-request pattern $\big\{(1, k_1), (2, k_2),\ldots, (n, k_n)\big\}$.
\STATE $j = 1, \lambda = 0$.
\WHILE{$\sum_{l = j}^{k}k_l  >  0$}
\IF{$k_j > 0$~and~the code symbol $x_j$ is not used to serve read requests for $x_1, x_2, \ldots, x_{j-1}$}
\STATE Serve one of the $k_j$ read requests for $x_j$ using the bucket storing the message symbol $x_j$. 
\STATE $\lambda = 1$. \LCOMMENT{{\it \small $\lambda = 1$ indicates that $1$ request for $x_j$ has been served.}}
\ENDIF 
\STATE Select $k_j - \lambda$ disjoint repair groups for $x_j$ that do not contain any code symbol used to serve the read requests for $x_1, x_2,\ldots, x_{j-1}$.
\STATE Map $k_j - \lambda$ (remaining) read requests for $x_j$ to the selected $k_j - \lambda$ disjoint repair groups.
\STATE $j = j + 1$.
\STATE $\lambda = 0$.
\ENDWHILE
\end{algorithmic}
\algrule[1pt]
\caption{Description of the algorithm to serve $k$ read requests by using disjoint sets of code symbols (buckets).}
\label{algo:batch}
\end{figure}

Towards this, we assume that for a $j  >  j^{\ast}$, $k_j > 0$ and the algorithm is able to map read requests $$\big\{(1, k_1), (2, k_2), \ldots, (j -1, k_{j-1})\big\}$$ to disjoint sets of code symbols. Let $\Sc_1, \Sc_2,\ldots, \Sc_{\sum_{u = 1}^{j-1}k_u}$ denote the disjoint sets of code symbols used by the algorithm to serve these $\sum_{u = 1}^{j-1}k_u$ read requests. Note that  the set $\Sc_{v}$, for $v \in [\sum_{u = 1}^{j-1}k_u]$, corresponds to a message symbol in $\{x_1, x_2,\ldots, x_{j-1}\}$ or a repair group of a message symbol in $\{x_1, x_2,\ldots, x_{j-1}\}$. It follows from Lemma~\ref{lem:no_6cycle} that the message symbol $x_j$ has at least $k$ disjoint repair groups. Moreover,  each of the sets $\Sc_1, \Sc_2,\ldots, \Sc_{\sum_{u = 1}^{j-1}k_u}$ has common code symbols with at most one repair group for the message symbol $x_j$. Therefore, the algorithm can find $k_j \leq (k -  \sum_{u = 1}^{j-1}k_u)$ repair groups for the message symbol $x_j$ which consist of code symbols that do not appear in the sets $\Sc_1, \Sc_2,\ldots, \Sc_{\sum_{u = 1}^{j-1}k_u}.$
\end{proof}

%% file: constructions.tex
\section{New Constructions of Multiset Batch Codes}
\label{sec:constructions}

We now utilize Theorem~\ref{thm:no_6cycle} to obtain new constructions of $(n, N, k, m)$ multiset batch codes with rate  $1 - o_k(1)$\footnote{Here, $b = o_k(a)$ implies that $\frac{b}{a} \to 0$ as $k \to \infty$}. Note that the requirement (i) in Theorem~\ref{thm:no_6cycle} translates to the graph being dense, while the requirement (ii) demands that the graph be not too desnse. Here, we also point out that for $\Cc$ to be an $(n, N, k, m = N)$ primitive multiset batch code, it is necessary for each node in $\Vc_1$ to have degree at least $k$ in the original graph $\Gc_{\Cc}$. In this section we take known explicit constructions of {\em dense} bipartite graphs with {\em girth at least $8$} and study the parameters of multiset  batch codes obtained from these graphs. These graphs strike a balance between the two counteracting requirements in Theorem~\ref{thm:no_6cycle}.

\subsection{Constructions based on the Balbuena graphs~\cite{balbuena2009}}
\label{sec:Balbuena_const}

Balbuena \cite{balbuena2009} proposed a method based on Latin Squares to construct the adjacency matrices of regular bipartite graphs of girth $8$. The construction of these graphs involves the concept of quasi row-disjoint matrices from {\cite{balbuena2008}}. The main result in \cite{balbuena2009} can be summarized as follows:

\begin{proposition}(\cite{balbuena2009})
\label{thm:balbuena}
Let $q$ be a prime power and $w$ be an integer such that $3\leq w \leq q$. Then, there exists an explicit $w$-regular balanced bipartite graph of girth $8$ on $2(wq^2 - q )$ nodes.
\end{proposition}

The bipartite graphs from \cite{balbuena2009} along with Theorem~\ref{thm:no_6cycle} yield  multiset batch codes with rate $1/2$ and the following parameters.

\begin{theorem}
For any prime power $k$, and a large enough integer $n$, there is an explicit $(n,N,k,m)$ Batch code with rate $\frac{1}{2}$, $m = 2k^3 - 2k$, and 
\[
N =  \ceilb{\frac{n}{k^3-k}} m
\]
\end{theorem}

\begin{proof}~
 Apply Proposition \ref{thm:balbuena} and set $w = q = k$ to obtain an $(n_1 = k^3 - k, N_1 = 2(k^3 - k), k, m_1 = N_1)$ primitive multiset batch code. Now, we can complete the proof by applying the gadget lemma.
\end{proof}

As we have seen, using the graphs from \cite{balbuena2009} directly,
we obtain multiset batch codes with rate $1/2$.
We can modify the graph  construction of \cite{balbuena2009}, and obtain
multiset batch codes with rate $1 - o_k(1)$ as follows.

\begin{theorem}\label{thm:balbuena2}
For any $\epsilon > 0$, any large enough  prime power $q$, and  large enough 
integer $n$, there is an explicit $(n,N,k,m)$ multiset batch code with rate 
at least $1 - \frac{1}{k^{\epsilon /4}}$, $m = k^{3+ \epsilon}
+ o(k^{3+ \epsilon})$, and 
\[
N =  \ceilb{\frac{n}{k^{3 + \epsilon}}} m
\]
where $k = \floorb{\frac{q}{b}}$, $b= \ceilb{(q^3 - q)^{\beta}}$ and
$\beta = \frac{1}{3} \frac{\epsilon}{4 + \epsilon}$.
\end{theorem}

\begin{proof}~
We start with the 4 and 6-cycle free regular 
bipatite graph $\Gc=(\Vc_1,\Vc_2,\Ec)$ obtained by 
 applying Proposition \ref{thm:balbuena} and setting  $w = q = k$.
Thus, in the graph $\Gc$, $|\Vc_1|=|\Vc_2|=q^3 - q$. We use $r$ to denote 
this number,
 that is $r=q^3-q$.
Recall that the degree of every vertex in $\Gc$ is $q$.

Next, we modify this construction, to obtain a new graph $\Gc'=(\Vc_1',\Vc_2',E')$
as follows.
Let $\beta = \frac{1}{3} \frac{\epsilon}{4 + \epsilon}$, 
and $b= \ceilb{(q^3 - q)^{\beta}}$.
To keep the presentation simple, from now on we assume that $b$ divides $q$.

In the new graph $\Gc'$, the number 
of vertices on the right side remains the same as in $\Gc$,
but we take $b$  ``copies'' of the left hand side vertices.
We distribute the edges adjacent to the left nodes of the graph
between the ``copies'', thus 
the  degree of every vertex in $\Vc_1'$ in $\Gc'$ will be 
$k = \frac{q}{b}$, and the degree of every vertex in $\Vc_2'$  
in $\Gc'$ remains $q$.
Note that $\beta < 1/3$ holds for any $\epsilon > 0$, thus $k>1$ holds
for large enough $q$.

More formally, let $|\Vc'_1|=b|\Vc_1|= br$, $\Vc'_2 = \Vc_2$, and $|\Ec'|=|\Ec|$. 
Thus each node in $\Vc_1$  corresponds to $b$ nodes in $\Vc'_1$
and we think of these as copies of the original node in $\Vc_1$.  
Let $\Vc'_1 = \{1,2, \cdots, br\}$. Then for $i'=1,2 \cdots, br$, 
we have $i'= ur +i$ for $u=0,1,\cdots, b-1$ and $i=1,2,\cdots, r$. 
Thus, for a given node $i\in \Vc_1$, we have the copies 
$i, r+i, 2r+i, \cdots, (b-1)r+i$ in $\Vc'_1$. 

We define the edges of $\Gc'$ as follows. 
Suppose that the neighbors of node $i\in \Vc_1$ in $\Gc$ are 
$\Nc(i)= \{j_1, j_2, \cdots, j_q\}$. 
Then the neighbors of node $ur+i \in \Vc'_1$ in $\Gc'$ are 
$N'(ur+i)=\{j_{u\frac{q}{b}+1},j_{u\frac{q}{b}+2}, 
\cdots,j_{u\frac{q}{b}+\frac{q}{b}}\}$. 
That is, each copy of a given node $i$ keeps a $\frac{1}{b}$ faction of its 
original neighbors (and each vertex in $\Vc_1'$ has $k=q/b$ neighbors).

Next we prove that $\Gc'$ has no 4-cycles or 6-cycles.

Assume that $\Gc'$ has  a 4-cycle, without loss of generality, 
suppose ${i'_1}$, ${i'_2}$, ${j_1}$, ${j_2}$ form a 4-cycle.
Notice that if 
\begin{equation}\label{eq:proof_no_4_cycle}
i'_1  \equiv i'_2 \pmod  r
\end{equation}
then in graph $\Gc$, there are two edges between the nodes ${i_1}$ and ${j_1}$, 
where $i_1 \equiv i'_1 \pmod r$,  $ 1\leq i_1\leq r  $.  
This is a contradiction, since $\Gc$ does not have parallel edges. 
Otherwise,  if (\ref{eq:proof_no_4_cycle}) does not hold, 
then there exists a 4-cycle (formed by ${i_1}, {i_2},{j_1},{j_2}$) 
in the graph $\Gc$, where $i_1 \equiv i'_1 \pmod r$, $i_2\equiv i'_2 \pmod r$,  
and $ 1\leq i_1,i_2\leq r  $. 
This is a contradiction, since $\Gc$ does not have 4-cycles.

Next, assume that $\Gc'$ has  a 6-cycle, without loss of generality, 
suppose ${i'_1}$, ${i'_2}$, ${i'_3}$, ${j_1}$, ${j_2}$, ${j_3}$ 
form a 6-cycle. Let us consider the following congruences:
\begin{equation}\label{eq:proof_no_6_cycle}
\begin{split}
i'_1 & \equiv i'_2 \pmod r\\
i'_2 & \equiv i'_3 \pmod r\\
i'_3 & \equiv i'_1 \pmod r\\
\end{split}
\end{equation}
If  none of the congruences (\ref{eq:proof_no_6_cycle}) holds, 
then there exists a 6-cycle in the graph $\Gc$. 
That 6-cycle is formed by 
${i_1}$, ${i_2}$, ${i_3}$, ${j_1}$, ${j_2}$, ${j_3}$, 
where $i_t \equiv i'_t \pmod  r $ and $1 \leq i_t \leq r$ for
$t=1,2,3$.
This is a contradiction, since $\Gc$ does not have 6-cycles. 
Otherwise, supose that  at least one of the congruences 
(\ref{eq:proof_no_6_cycle}) holds. 
Without loss of generality, assume that $i'_1  \equiv i'_2 \pmod  r$. 
Since ${i'_1}$, ${i'_2}$, ${i'_3}$, ${j_1}$, ${j_2}$, ${j_3}$ 
form a 6-cycle, one of $j_1,j_2,j_3$, say $j_1$ is connected to both 
$i'_1$ and $i'_2$. Then there are two edges between ${i_1}$ and ${j_1}$ in 
the graph $\Gc$, where $i_1 \equiv i'_1 \pmod r$ and $ 1\leq i_1 \leq r  $. 
This is a contradiction, since $\Gc$ does not have parallel edges.

Thus, we have shown that $\Gc'$ does not have 4-cycles or 6-cycles.
Then, by Theorem~\ref{thm:no_6cycle} we obtain a primitive multiset
batch code with number of servers $m=|\Vc_1'| + |\Vc_2'|= (b+1)r$, 
$k=q/b$, and rate 
$\frac{b}{b+1}= 1 - \frac{1}{b+1}> 1 - \frac{1}{b}$.
Substituting $b= \ceilb{(q^3 - q)^{\beta}}$ where
$\beta = \frac{1}{3} \frac{\epsilon}{4 + \epsilon}$, and expressing
both $m$ and the rate as a function of $k$ 
we get 
$$m = k^{3 + \frac{12 \beta}{1 - 3 \beta}} + k^{3 + 
\frac{9 \beta}{1 - 3 \beta}} 
= k^{3+ \epsilon} + o(k^{3+ \epsilon})$$ 
and  rate 
at least $$ 1 - \frac{1}{k^{\frac{3 \beta}{1 - 3 \beta}}}
=1 - \frac{1}{k^{\epsilon /4}}.$$
Applying the Gadget Lemma with $g = \ceilb{\frac{n}{k^{3 + \epsilon}}}$ we obtain an $(n, N, k, m)$ multiset batch code with
\[
N =  \ceilb{\frac{n}{k^{3 + \epsilon}}} m 
\]
and rate at least $1 - \frac{1}{k^{\epsilon /4}}$.

\end{proof}

\begin{theorem} \label{thm:balbuena3}
For any integer  $c \geq 2$, 
any large enough  prime power $q$, and  large enough 
integer $n$, there is an explicit $(n,N,k,m)$ multiset batch code with rate 
$1 - \frac{1}{c}$, $m \leq c (c-1)^3 k^3 = O(k^3) $ and
\[
N \leq  \ceilb{\frac{n}{(c-1)(q^3 - q)}} m \leq
\ceilb{\frac{1.1n}{(c-1)^4 k^{3}}} m = O(\frac{n}{k^3} m)
\]
where $k = \floorb{\frac{q}{b}}$ and $b= c-1$.
\end{theorem}

\begin{proof}~
Substitute $k = \floorb{\frac{q}{b}}$ and $b= c-1$ in the above construction.
\end{proof}

\subsection{Constructions based on the Lazebnik, Ustimenko, Woldar graphs~\cite{lazebnik}}
\label{sec:Lazebnik_const}

Here, we present the construction of  bipartite graphs with girth at least $8$ from~\cite{lazebnik}. These graphs allow for a left degree $k$  which is polynomial in the number of left nodes  $n$ while achieving the rate $1 - o_{k}(1)$. First, we briefly describe the construction of these bipartite graphs and then comment on the parameters of the multiset batch codes obtained based on these graphs.

\textbf{Lazebnik {\em et al.} construction :} 
Given an odd prime power $q$, let $t \in (0, 2]$ and $s \in (0, 1]$ be such that $q^t$ and $q^s$ are integers. The bipartite graph $\Bc_{s, t}(q) = \Gc = (\Vc_1, \Vc_2, \Ec)$ with $q^{3 + t}$ left nodes and $q^{3 + s}$ right nodes is constructed as follows. Let $\Tc \subseteq \F_{q^2}$ and $\Sc \subseteq \F_q$ with $|\Tc| = q^t$ and $|\Sc| = q^{s}$. Associate $q^{3 + t}$ left nodes with $3$-dimensional vectors such that 
$$
\Vc_{1} = \{\lv = (l_1, l_2, l_3)~:~l_1 \in \Tc \subseteq \F_{q^2}, l_2 \in \F_{q^2}, l_3 \in \F_{q}\}.
$$
The $q^{3 + s}$ right nodes of the bipartite graph $\Gc$ are denoted by $3$-dimensional vectors such that 
$$
\Vc_{2} = \{\vv = (v_1, v_2, v_3)~:~v_1 \in \Sc \subseteq \F_q, v_2 \in \F_{q^2}, v_3 \in \F_{q}\}.
$$
In the graph $\Gc$, there exists an edge between the left node $\lv = (l_1, l_2, l_3)$ and the right node $\vv = (v_1, v_2, v_3)$, {\em i.e.}, $(\lv, \vv) \in \Ec$, iff we the following two conditions are satisfied
\begin{subequations}
\begin{align}
l_2 - v_2 &= l_1v_1, \\
l_3 - v_3 &= f(l_1)v_2 + l_1f(v_2).
\end{align}
\end{subequations}
Here, $f: x \mapsto f(x)$ represents the involutive automorphism of $\F_{q^2}$ with fixed field $\F_q$. Lazebnik {\em et al.} establish that the graph $\Bc_{s, t}(q)$ has girth at least $8$~\cite[Theorem 3]{lazebnik}. Here, we paraphrase the result to make it consistent with the rest of the paper. 

\begin{proposition}{\cite[Theorem 3]{lazebnik}}
\label{prop:lazebnik}
Let $q$ be an odd prime power. Then for any $t \in (0, 2]$ and $s \in (0, 1]$ such that both $q^t$ and $q^s$ are integers, the bipartite graph $\Bc_{s, t}(q)$ is biregular with $q^{3+s} + q^{3+t}$ total nodes ($q^{3+t}$ left nodes, and $q^{3+s}$ right nodes), $q^{3+s+t}$ edges and girth at least $8$. The left degree and the right degree of  $\Bc_{s, t}(q)$ are $q^s$ and $q^t$, respectively.
\end{proposition}

\begin{remark}
For the case when $\Tc = \F_{q^2}$ and $q \geq 3$, Lazebnik {\em et al.} show that the graph $\Bc_s(q) = \Bc_{s, 2}(q)$ has girth exactly equal to $8$.
\end{remark}

We now present the main result of this subsection, which states the parameters of the multiset batch codes obtained from the bipartite graphs in \cite{lazebnik}.
\begin{theorem}
\label{thm:lazebnik}
For any $0 < \epsilon \leq 1$, one can find an infinite sequence of  odd prime powers $k$ such that for large enough integer $n$, there is an explicit $(n, N, k, m)$ multiset batch code with rate at least  $1 - \frac{1}{k^{\epsilon}}$, $m =  k^{4 + \tilde{\epsilon}} + o( k^{4 + \tilde{\epsilon}})$, and $$N = \ceilb{\frac{n}{k^{4 + \tilde{\epsilon}}}}m.$$ Here, $\tilde{\epsilon}$ is a rational number which can be taken arbitrarily close to $\epsilon$. 
\end{theorem}

\begin{proof}~
 From Proposition~\ref{prop:lazebnik}, we know that for any $t \in (1, 2]$ such that $q^t$ is an integer, the bipartite graph $\Gc = \Bc_{s = 1, t}(q)$ obtained by Lazebnik {\em et al.} construction has $q^{3 + t}$ left nodes, $q^{3 + s} = q^{4}$ right nodes, and girth at least $8$. Also, note that each left node in $\Bc_{s = 1, t}(q)$ has degree $q^s = q$. Therefore, it follows from Theorem~\ref{thm:no_6cycle} that the bipartite graph $\Bc_{s = 1, t}(q)$ gives an $(n_1 = q^{3 + t}, N_1 = q^{3 + t} + q^4, k = q, m = N_1)$ primitive batch code $\Cc$. The rate of $\Cc$ is $\frac{q^{3 + t}}{q^{3 + t} + q^4} \geq 1 - \frac{q^{4}}{q^{3 + t}} = 1 - \frac{1}{q^{t-1}}$.

Now, for an $\epsilon > 0$, we pick a rational number $\tilde{\epsilon}$ such that $\epsilon \leq \tilde{\epsilon}$ for $\tilde{\epsilon}$ arbitrarily close to $\epsilon$.  By setting $t = 1 + \tilde{\epsilon}$ and choosing $k = q$ to be an odd prime power so that $q^{1 + \tilde{\epsilon}}$ is an integer\footnote{Note that for any rational number $\tilde{\epsilon}$ there is an infinite sequence of such odd prime powers.}, we have that $\Cc$ is an $(n_1 = k^{4 + \tilde{\epsilon}}, N_1 = k^{4 + \tilde{\epsilon}} + k^4, k, m = N_1)$ primitive batch code with rate at least $1 - \frac{1}{k^{\tilde{\epsilon}}} \geq 1 - \frac{1}{k^{{\epsilon}}}$. We now apply the gadget lemma (see Section~\ref{sec:prelim}) with $g = \ceilb{\frac{n}{n_1}} = \ceilb{\frac{n}{k^{4 + \tilde{\epsilon}}}}$ to obtain an $(n, N, k, m)$ batch code with rate at least  $1 - \frac{1}{k^{\epsilon}}$, $m =  k^{4 + \tilde{\epsilon}} + o( k^{4 + \tilde{\epsilon}})$, and $$N = \ceilb{\frac{n}{k^{4 + \tilde{\epsilon}}}}m.$$
\end{proof}

%% file: constructions2.tex
\subsection{Constructions based on the de Caen, Sz\'ekely graphs}
\label{sec:deCaen_const}

In \cite{decaen},  de Caen and Sz\'ekely used known constructions of 
generalized quadrangles to construct  bipartite graphs
with maximum number of edges and no $4$- and $6$-cycles. Generalized quadrangles are combinatorial structures defined as follows.
Let $\Pc$ be a set of {\em points}, and let $\mathcal{L}$ be a collection 
of subsets of $\Pc$ called {\em lines}.
$(\Pc, \mathcal{L})$ is called a {\em generalized quadrangle} of order 
$(s,t)$ ($s,t \geq 1$) if the following holds:

\begin{enumerate}
\item Each point is incident with $t+1$ lines, and two distinct points
are incident with at most one line.

\item Each line is incident with $s+1$ points, and two distinct lines are 
incident with at most one point.

\item If $x \in V$ is a point and $L \in \mathcal{L}$ is a line
not incident with $x$, then there is a unique pair $(y, M)$,
where $M$ is a line and $y$ is a point incident to $M$, such that
$x$ is incident to $M$ and $y$ is incident to $L$.
\end{enumerate}

The conditions above imply that the number of points $|\Pc|=(s+1)(st+1)$
and the number of lines $|\mathcal{L}|=(t+1)(st+1)$. Given any generalized quadrangle, we can represent it by a bipartite graph
$\Gc=(\Vc_1,\Vc_2, \Ec)$ in the straightforward way, by letting $\Vc_1 = \mathcal{L}$,
$\Vc_2=\Pc$ and $(L,x) \in \Ec$ if and only if $x$ is incident to $L$.
Note that the number of edges of this graph is $|\mathcal{L}|(s+1) = 
|\Pc|(t+1)=(s+1)(t+1)(st+1)$. 

The first two requirements in the definition of generalized quadrangles 
immediately imply that the bipartite graph representing a generalized
quadrangle cannot contain any 4-cycles, and the last condition immediately
implies that the graph cannot contain any 6-cycles. De Caen and Sz\'ekely observed that the bipartite graphs representing
generalized quadrangles in fact contain the largest possible
number of edges among graphs on the same number of vertices 
without $4$- and $6$-cycles.

Generalized quadrangles have been extensively studied
and several constructions are known for settings where $s$ and $t$ are
close to each other, e.g when $s=t$ or $s=q-1$ and $t=q+1$.
However, to obtain multiset batch codes
with rate close to 1, we need generalized quadrangles where $s$ and $t$
have different orders of magnitude.
There are two families of constructions of generalized quadrangles 
of this type known: for any prime power $q$, there are 
generalized quadrangles of order $(q,q^2)$ and of order $(q^2,q^3)$.
For more details on generalized quadrangles see \cite{quadbook1, quadbook2}. 

Generalized quadrangles of order $(q,q^2)$ yield multiset batch codes
with $m = (1 + o_k(1)) k^5$ and rate $1 - 1/(k-1)$.
The performance of these codes is similar to the codes based on the 
Lazebnik et al. construction when taking $\epsilon =1$ in that construction.

Generalized quadrangles of order $(q^2,q^3)$ yield multiset batch codes
with $m = (1+ o_k(1)) k^4$ and rate $1 - 1/\sqrt{k}$.

\begin{theorem}\label{thm:decaen}
Let $q$ be any prime power, and let $k=q^2$.
Then for any large enough integer $n$ 
there is an explicit $(n,N,k,m)$ multiset batch code with
rate at least $1 - 1/\sqrt{k}$, $m = (1 + o_k(1)) k^4$, 
and $N = \lceil \frac{n}{k^4} \rceil m$.
\end{theorem}

\begin{proof}~
The known constructions of generalized quadrangles of order $(q^2,q^3)$ 
immediately give bipartite graphs with $n_1 = |\Vc_1| = (q^3 +1)(q^5 +1)$,
$|\Vc_2|=(q^2+1)(q^5 +1)$, and left degree $k = q^2+1$.
We also know that these bipartite graphs have no $4$- and $6$-cycles.
Using Theorem~\ref{thm:no_6cycle}, this gives $(n_1,N_1,k,m)$ primitive multiset
batch codes with $N_1 = m = |\Vc_1| + |\Vc_2| = q^8  +  o_k(q^{8})=  (1 + o_k(1))k^4$.
The rate of the resulting primitive multiset batch code is 
$|\Vc_1|/(|\Vc_1| + |\Vc_2|) \geq  1 - 1/q = 1 - 1/\sqrt{k}$.
Applying the gadget lemma as in the previous subsection gives the
statement of the theorem.
\end{proof}

%% file: zigzag_construction.tex
\subsection{Constructions based on zig-zag codes of~\cite{zigzag13}}
\label{sec:zigzag} 

In this subsection, we describe a construction for  primitive mutliset batch codes that allow for $m = \Theta(k^{k^{1 + \epsilon} + 1 + \epsilon})$ with any $0 < \epsilon$ while achieving rate at least $1 - (1/{k^{\epsilon}})$. In particular, we present a construction of  biregular bipartite graphs which are free of $4$- and $6$-cycles. For a given prime $k$ and an integer $r$, the graph has $rk^{r}$ left nodes and $k^{r+1}$ right nodes. The left degree and right degree of these graphs are $k$ and $r$, respectively. The construction presented here is based on the work of Tamo {\em et al.} on zig-zag codes~\cite{zigzag13}. This construction has been  previously employed in \cite{availability}, where it gives codes with disjoint repair groups, {\em i.e.}, the bipartite graph is free of $4$-cycles. Here, we go one step further and show that the bipartite graphs obtained from zig-zag construction are free of $6$-cycles as well. It then follows from Theorem~\ref{thm:no_6cycle} that these bipartite graphs give multiset batch codes. 

\textbf{Zig-zag construction  :} 
Let $k$ be a prime number\footnote{The construction also works when $k$ is a prime power. However, for the ease of exposition, we only present the construction when $k$ is a prime number.} and $r$ be an integer. Each of the $rk^r$ left nodes of the bipartite graph are represented by an $(r+1)$-dimensional vectors such that 
$$
\Vc_{1} = \{\lv = (l_0, l_1, l_2,\ldots, l_r)~:~l_0 \in [r]~\text{and}~l_i \in [0:(k-1)]~\text{for}~i\in [r]\}.
$$
Similarly, $k^{r+1}$ right nodes of the bipartite graphs are denoted by an $(r +1 )$-dimensional vectors such that
$$
\Vc_{2} = \{\vv = (v_0,v_1, v_2,\ldots, v_r)~:~v_i \in [0:(k-1)]~\text{for}~i \in [0:r]\}.
$$
Given a left node $\lv = (l_0, l_1, l_2,\ldots, l_r)$ and a right node $\vv = (v_0,v_1, v_2,\ldots, v_r)$, there is an edge between $\lv$ and $\vv$, {\em i.e.}, $(\lv, \vv) \in \Ec$  iff
\begin{align}
\label{eq:edge_zigzag}
(l_1, l_2,\ldots, l_r) + v_0\ev_{l_0} = (v_1, v_2,\ldots, v_r)~({\rm mod}~k).
\end{align}
Here, $\ev_{l_0} \in \mathbb{Z}^r$ denotes the standard unit vector which has $1$ at the $l_0$-th coordinate and zeros at other $r-1$ coordinates. From \eqref{eq:edge_zigzag}, a given left node $\lv$, has $k$ neighbors for $k$ values of $v_0$. Similarly, a right node $\vv$ has $r$ left neighbors as for each $l_0 \in [r]$ there is a unique vector $(l_1, l_2,\ldots, l_r) \in [0:(k-1)]^r$ such that \eqref{eq:edge_zigzag} holds.

In \cite[Claim 1]{availability}, it is shown that the bipartite graph obtained using zig-zag construction  has no $4$-cycle\footnote{The Claim $1$ in \cite{availability} states this result in an alternative manner.}. Here, we state this result without its proof. 

\begin{lemma}{\cite{availability}}
\label{lem:4_zigzag}
The biregular bipartite graph obtained by zig-zag construction  has no $4$-cycle.
\end{lemma}

Next, we establish that the bipartite graphs generated by zig-zag construction  are free of $6$-cycles as well. 

\begin{lemma}
\label{lem:girth_zigzag}
Let $\Gc$ be a biregular bipartite graph obtained from zig-zag construction . Then,  the girth of $\Gc$ is at least $8$. 
\end{lemma}

\begin{proof}~
We have from Lemma~\ref{lem:4_zigzag} that $\Gc$ is free of $4$-cycles. Here, we show that it is not possible to have a $6$-cycle in $\Gc$ which establishes the Lemma~\ref{lem:girth_zigzag}. In order to show contradiction, we assume that there is a $6$-cycle $\{\lv^1, \vv^1, \lv^2, \vv^2, \lv^3, \vv^3\}$ in the graph $\Gc$. Here, $\{\lv^1,  \lv^2,  \lv^3\}$ and $\{\vv^1, \vv^2, \vv^3\}$ represent three distinct left and right nodes, respectively. From \eqref{eq:edge_zigzag} we have the following relations.
\begin{subequations}
  \begin{align}
 \label{eq:block1a}
(l^1_1, l^1_2,\ldots, l^1_r) + v^1_0\ev_{l^1_0} = (v^1_1, v^1_2,\ldots, v^1_r)~({\rm mod}~k),\\
(l^2_1, l^2_2,\ldots, l^2_r) + v^1_0\ev_{l^2_0} = (v^1_1, v^1_2,\ldots, v^1_r)~({\rm mod}~k).
 \label{eq:block1b}
  \end{align}
\end{subequations}
\begin{subequations}
  \begin{align}
 \label{eq:block2a}
(l^2_1, l^2_2,\ldots, l^2_r) + v^2_0\ev_{l^2_0} = (v^2_1, v^2_2,\ldots, v^2_r)~({\rm mod}~k),\\
(l^3_1, l^3_2,\ldots, l^3_r) + v^2_0\ev_{l^3_0} = (v^2_1, v^2_2,\ldots, v^2_r)~({\rm mod}~k).
 \label{eq:block2b}
  \end{align}
\end{subequations}
and
\begin{subequations}
  \begin{align}
 \label{eq:block3a}
(l^3_1, l^3_2,\ldots, l^3_r) + v^3_0\ev_{l^3_0} = (v^3_1, v^3_2,\ldots, v^3_r)~({\rm mod}~k), \\
(l^1_1, l^1_2,\ldots, l^1_r) + v^3_0\ev_{l^1_0} = (v^3_1, v^3_2,\ldots, v^3_r)~({\rm mod}~k).
 \label{eq:block3b}
  \end{align}
\end{subequations}
By subtracting \eqref{eq:block3b} from \eqref{eq:block1a} we obtain
\begin{align}
\label{eq:i1}
 (v_0^1 - v^3_0)\ev_{l^1_0} = (v^1_1 - v^3_1, v^1_2 - v^3_2,\ldots, v^1_r - v^3_r)~({\rm mod}~k).
\end{align}
We need to have $v^1_0 \neq v^3_0$, otherwise from \eqref{eq:i1} we have $\vv^1 = \vv^2$, which is not true as $\vv^1$ and $\vv^2$ are distinct nodes. By subtracting \eqref{eq:block1b} from \eqref{eq:block1a}, we obtain
 \begin{align}
\label{eq:v1}
(l^1_1 - l^2_1, l^1_2 - l^2_2,\ldots, l^1_r - l^2_r) + v_0^1(\ev_{l^1_0} - \ev_{l^2_0}) = 0.
\end{align}
Again, we need to have $l^1_0 \neq l^2_0$. Otherwise, \eqref{eq:v1} implies that $\lv^1 = \lv^2$, which is not true. Similarly, we can obtain the following sets of equations
\begin{align}
\label{eq:i2}
 (v_0^1 - v^2_0)\ev_{l^2_0} = (v^1_1 - v^2_1, v^1_2 - v^2_2,\ldots, v^1_r - v^2_r)~({\rm mod}~k), \\
 (v_0^2 - v^3_0)\ev_{l^3_0} = (v^2_1 - v^3_1, v^2_2 - v^3_2,\ldots, v^2_r - v^3_r)~({\rm mod}~k).
\label{eq:i3}
\end{align}
and
\begin{align}
\label{eq:v2}
(l^2_1 - l^3_1, l^2_2 - l^3_2,\ldots, l^2_r - l^3_r) + v_0^2(\ev_{l^2_0} - \ev_{l^3_0}) = 0, \\
(l^3_1 - l^1_1, l^3_2 - l^1_2,\ldots, l^3_r - l^1_r) + v_0^3(\ev_{l^3_0} - \ev_{l^1_0}) = 0.
\label{eq:v3}
\end{align}

Combining \eqref{eq:v1}, \eqref{eq:v2} and \eqref{eq:v3}, we obtain that $l^1_0, l^2_0$ and $l^3_0$ are three distinct integers in $[r]$. Similarly, it follows from \eqref{eq:i1}, \eqref{eq:i2}, and \eqref{eq:i3} that $v^1_0, v^2_0$ and $v^3_0$ are distinct integers. Moreover, we have from \eqref{eq:i1}, \eqref{eq:i2}, and \eqref{eq:i3} that
\begin{subequations}
 \begin{align}
 \label{eq:s1a}
 v^1_ {l^1_0} &\neq v^3_ {l^1_0}\\
 v^1_ {z} &= v^3_ {z},~\text{for}~z \in [r]\backslash\{l^1_0\}.
 \label{eq:s1b}
  \end{align}
\end{subequations}
\begin{subequations}
  \begin{align}
 \label{eq:s2a}
 v^1_ {l^2_0} &\neq v^2_ {l^2_0}\\
 v^1_ {z} &= v^2_ {z},~\text{for}~z \in [r]\backslash\{l^2_0\}.
 \label{eq:s2b}
  \end{align}
\end{subequations}
and
\begin{subequations}
  \begin{align}
 \label{eq:s3a}
 v^2_ {l^3_0} &\neq v^3_ {l^3_0}\\
 v^2_ {z} &= v^3_ {z},~\text{for}~z \in [r]\backslash\{l^3_0\}.
 \label{eq:s3b}
  \end{align}
\end{subequations}
Note that $l^1_0$, $l^2_0$, and $l^3_0$ are distinct integers. Therefore, the equation \eqref{eq:s1a} to \eqref{eq:s3b} lead to contradiction as it is not possible to simultaneously satisfy all of these equations. Thus, there is no $6$-cycle in the graph $\Gc$.
\end{proof}

The following result presents the parameters of the multiset batch codes obtained from zig-zag construction . 
\begin{theorem}
\label{thm:zigzag}
For any fixed integer $k$, a large enough integer $n$, and $\epsilon > 0$, there is an explicit $(n, N, k, m)$ multiset batch code with rate at least $1 - \frac{1}{k^{\epsilon}}$, $m = \Theta(k^{k^{1 + \epsilon} + 1 + \epsilon})$, and $$N = \ceilb{\frac{n}{k^{k^{1 + \epsilon} + 1 + \epsilon}}}m.$$
\end{theorem}

\begin{proof}~
 From Lemma~\ref{lem:girth_zigzag}, we know that a bipartite graph $\Gc$ obtained by zig-zag construction  has $rk^r$ left nodes and $k^{r + 1}$ right nodes, and girth at least $8$. Moreover, the degree of each left node in $\Gc$ is $k$. Therefore, it follows from Theorem~\ref{thm:no_6cycle} that $\Gc$ provides an $(n_1 = rk^{r}, N_1 = (r + k)k^{r}, k, m = N_1)$ primitive batch code $\Cc$ with rate $\frac{rk^r}{rk^r + k^{r + 1}} = \frac{r}{r + k} \geq 1 - \frac{k}{r}$. 

We now select $r = k^{1 + \epsilon}$, which gives us the rate of the code at least $1 - \frac{1}{k^{\epsilon}}$\footnote{If $k^{1+ \epsilon}$ is not an integer, $r$ can be made integer by taking the floor or the ceiling of $k^{1+ \epsilon}$.}. This particular choice for the parameter $r$ implies that $\Cc$ is an $(n_1 = k^{k^{1+\epsilon} + 1+ \epsilon}, N_1 = (k^{\epsilon} + 1)k^{k^{1 + \epsilon}+ 1}, k, m = N_1)$ primitive batch code. Now using gadget lemma (see Section~\ref{sec:prelim}) with $g = \ceilb{\frac{n}{n_1}} =  \ceilb{\frac{n}{k^{k^{1 + \epsilon} + 1 + \epsilon}}}$, we obtain an $(n, N, k, m)$ batch code $g\cdot\Cc$ with rate at least $1 - \frac{1}{k^{\epsilon}}$, $
m = \Theta(k^{k^{1 + \epsilon} + 1 + \epsilon})$, and 
$$N = \ceilb{\frac{n}{k^{k^{1 + \epsilon} + 1 + \epsilon}}}m.$$
\end{proof}

\subsection{Limitations of the bipartite graph based constructions}\label{limit}

We note that the bipartite graph based constructions described above 
exhibit a tradeoff between the parameter $m$ as a function of $k$,
 and the rate of the codes
obtained as a function of $k$. 
Ideally, we would like to minimize $m$ (the number of servers) as a function 
of $k$ (the number of read requests
that can be supported) and maximize the rate as a function of $k$.

Our constructions achieve close to optimal tradeoffs by 
the following estimates of de Caen and Sz\'ekely \cite{decaen}.

\begin{theorem}{\cite{decaen, decaen2}}\label{estim}
Let $G=(\Vc_1, \Vc_2, \Ec)$ be a bipartite graph without 4- and 6-cycles,
where $|\Vc_1|=n_1$, $|\Vc_2|=n_2$ and $\sqrt{n_1} \leq n_2 \leq n_1$.
Then $|\Ec| = O(n_1^{2/3}n_2^{2/3})$.
\end{theorem}

Note that we can assume $n_2 \geq \sqrt{n_1}$, since otherwise
$|\Ec|= O(n_1)$ \cite{decaen2} which would only allow a constant number
of read requests $k$.

These estimates imply that assuming rate at least $1/2$ the number of servers
$m$ has to be at least $\Omega(k^{3})$
in multiset batch codes based on bipartite graphs without
4- and 6-cycles. 
More precisely, assuming rate at least $1 - \frac{1}{k^{\alpha}}$
implies $m = \Omega(k^{3 + 2 \alpha})$. 
Thus, our results give close to optimal tradeoffs for 
multiset batch codes based on bipartite graphs.

We also note that any systematic primitive multiset batch code where 
decoding uses at most 
one non-systematic code symbol for every request (as in our code
constructions) must have an associated bipartite graph without 4-
and 6-cycles. 
Thus, our results in Sections~\ref{sec:Balbuena_const} and \ref{sec:deCaen_const} give close to 
optimal tradeoffs between the number of servers and the rate for multiset batch codes of this type.

%% file: reliability.tex
\section{On Fault Tolerance of Batch Codes}
\label{sec:distance} 

In~\cite{batch}, Ishai {\em et al.} study batch codes with {\em load balancing} as the primary focus of their work. However, in practice the codes for distributed storage systems need to allow for mechanisms to protect data against catastrophic events. These catastrophic events may correspond to the permanent or persistent loss of a large number of servers due to multiple reasons, including hardware failures, power outages, and loss of rack switches in data centers~\cite{ford2010, DeanSoCC}. Therefore, it is desirable to work with codes that have large minimum distance. Here, we note that Lipmaa and Skachek have recently shown that a binary linear batch code which supports any $k$-request patterns has minimum distance at least $k$~\cite{LS14}.

Here, we describe a straightforward yet powerful way to combine the notion of load balancing in batch codes with the issue of fault tolerance. We show that one can modify a given batch code to get a code with large minimum distance at the cost of rate loss. In particular, given a field of large enough size and a batch code with rate $1 - o_k(1)$, one can get a code which is also a batch code and has minimum distance arbitrarily close to the Singleton bound.

\begin{lemma}
\label{lem:dist_composition}
Let $\Cc^{\rm G}$ be an $[n, N^{\rm G}, d^{\rm G}]$\footnote{An  $[n, N, d]$ code is an $[n, N]$ code with minimum distance $d$.} systematic linear code with minimum distance $d^{\rm G}$, and $\Cc^{\rm B}$ be an $(n, N^{\rm B}, k, m^{\rm B} = N^{\rm B})$ systematic primitive batch code. Then, it is possible to obtain an $(n, N = N^{\rm B} +  N^{\rm G} - n, k, m = N)$ systematic primitive batch code  $\Cc = \Cc^{\rm G} \circ \Cc^{\rm B}$ with distance at least $d^{\rm G}$.
\end{lemma}


\begin{proof}~
Let $\Cc^{\rm G}$ be an $[n, N^{\rm G} = n + p, d^{\rm G}]$ systematic linear code with minimum distance $d^{\rm G}$ which encodes a length-$n$ message vector $\xv = (x_1, x_2,\ldots, x_n) \in \Sigma^n$ to a codeword $$\cv^{\rm G} = (c^{\rm G}_1= x_1, c^{\rm G}_2 = x_2,\ldots, c^{\rm G}_n = x_n, c^{\rm G}_{n+1},\ldots, c^{\rm G}_{N^{\rm G}}) \in \Sigma^{N^{\rm G}}.$$ Let $\Cc^{\rm B}$ be an $(n, N^{\rm B}, k, m = N^{\rm B})$ systematic primitive batch code which encodes the length-$n$ message vector $(x_1, x_2,\ldots, x_n) \in \Sigma^n$ to a codeword $$\cv^{\rm B} = (c^{\rm B}_1 = x_1, c^{\rm B}_2 = x_2,\ldots, c^{\rm B}_n = x_n, c^{\rm B}_{n+1},\ldots,  c^{\rm B}_{N^{\rm B}}).$$ Given $\Cc^{\rm B}$ and $\Cc^{\rm G}$, we obtain a systematic code $\Cc = \Cc^{\rm G}\circ\Cc^{\rm B}$ which encodes the length-$n$ message vector $(x_1, x_2,\ldots, x_{n}) \in \Sigma^{n}$ to a codeword 
$$
\cv = (c_1, c_2,\ldots, c_{N = N^{\rm B} + p}) = (x_1, x_2,\ldots, x_n, c^{\rm B}_{n + 1}, \ldots, c^{\rm B}_{N^{\rm B}}, c^{\rm G}_{n + 1},\ldots, c^{\rm G}_{N^{\rm G}}) \in \Sigma^{N^{\rm B} +  N^{\rm G} - n}.
$$
It is clear from the construction that the minimum distance of $\Cc$ is at least as large as the minimum distance of $\Cc^{\rm G}$ as the codewords of $\Cc^{\rm  G}$ are sub-vectors of the codewords of $\Cc$. Similarly, one can obtain the codeword corresponding to $\xv$ in $\Cc^{\rm B}$ by puncturing the codeword associated with $\xv$ in $\Cc$. Therefore, $\Cc$ is a batch code that can support sequences of $k$ read requests.

\end{proof}

\begin{remark}
\label{re:dmin_batch}
Note that Lemma~\ref{lem:dist_composition} holds for multiset batch codes as well. Also, we state Lemma~\ref{lem:dist_composition} only for primitive batch codes. One can employ the gadget lemma (see~Section~\ref{sec:prelim}) with primitive multiset batch code $\Cc = \Cc^{\rm G} \circ \Cc^{\rm B}$ to obtain general multiset batch codes which can tolerate the failure of any $d^{\rm G} - 1$ servers (buckets).
\end{remark}

Taking $\Sigma$ to be a field with large enough size and $\Cc^{\rm G}$ to a systematic MDS code, we obtain the following result. Note that if $\Sigma$ is not sufficiently large, we can create parities over the extension field $\Sigma^{\rm B}$.

\begin{lemma}
\label{lem:MDS_batch}
For large enough integers $k$ and $n$, let $\Cc^{\rm B}$ be an $(n, N^{\rm B}, k, m^{\rm B} =  N^{\rm B})$ systematic primitive batch code with rate at least $1 - o_{k}(1)$. Further assume that $\Sigma$ is a field with large enough size. Then for any rate $0 < \rho < 1$, it is possible to obtain an $(n, N, k, m = N)$ primitive batch code $\Cc$ with minimum distance $d \geq (N - n + 1) - o_{k}(n)$.
\end{lemma}
\begin{proof}~
From the assumption on the rate of $\Cc^{\rm B}$, we know that 
\begin{align}
\label{eq:nN}
n \geq (1 - o_{k}(1))N^{\rm B}. 
\end{align}
Now, we take $\Cc^{\rm G}$ in Lemma~\ref{lem:dist_composition} to be an $[n, N^{\rm G}, N^{\rm G} - n + 1]$ systematic MDS code, which gives us an $(n, N = N^{\rm B} +  N^{\rm G} - n, k, m = N)$ systematic primitive batch code with minimum distance
$$
d \geq N^{\rm G} - n + 1.
$$
The above expression can be rewritten as 
\begin{align}
\label{eq:MDS_batch}
d &\geq N + n -  N^{\rm B} - n + 1 \nonumber \\
& = N - n + 1 - (N^{\rm B} - n) \nonumber \\
& \overset{(i)}{\geq} N - n + 1 - o_{k}(1)N^{\rm B} \nonumber \\
&\overset{(ii)}{=} N - n + 1 - o_{k}(n).
\end{align}
Here, $(i)$ and $(ii)$ follow from \eqref{eq:nN}. The rate $\rho$ of $\Cc$ is equal to $\frac{\rho^{\rm G}}{1 + o_k(1)}$, where $\rho^{\rm G} = \frac{n}{N^{\rm G}}$ denotes the rate of $\Cc^{\rm G}$. Note that we can select $\rho^{\rm G}$ to be any number in the interval $(0, 1)$.
\end{proof}